\documentclass{article}

\usepackage[preprint]{neurips_2025}
\usepackage{makecell}

\usepackage[utf8]{inputenc} 
\usepackage[T1]{fontenc}    
\usepackage{hyperref}       
\usepackage{url}            
\usepackage{booktabs}       
\usepackage{amsfonts}       
\usepackage{nicefrac}       
\usepackage{microtype}      
\usepackage{xcolor}         
\usepackage{multirow}
\usepackage{floatrow}

\newcommand{\ourmethod}{Adaptive Beam Search}
\usepackage{amsmath} 
\usepackage{amsfonts}   
\usepackage{amsthm}  
\usepackage{cleveref}
\usepackage{float}  
\usepackage{algorithm}
\usepackage[noend]{algpseudocode}
\usepackage{tikz}
\usepackage{placeins}
\usepackage{wrapfig}
\usepackage{subcaption}
\usepackage{hyperref}  
\usetikzlibrary{decorations.markings}

  \usepackage{nth}
  \usepackage{intcalc}

  \newcommand{\cAAAI}[1]{AAAI\ Conference\ on\ Artificial (AAAI)}

\newcommand{\E}{\mathbb{E}}

\newcommand{\R}{\mathbb{R}}  

\DeclareMathOperator{\argmin}{argmin}
\DeclareMathOperator{\argmax}{argmax}

\newcommand{\bv}[1]{\mathbf{#1}}

\algrenewcommand\algorithmicrequire{\textbf{Input:}}
\algrenewcommand\algorithmicensure{\textbf{Output:}}
\newcommand{\algrule}[1][.2pt]{\par\vskip.2\baselineskip\hrule height #1\par\vskip.2\baselineskip}

\newtheorem{theorem}{Theorem}
\newtheorem{definition}{Definition}

\newtheorem{claim}[theorem]{Claim}

\usepackage{enumitem}

\usepackage{natbib}

\title{Distance Adaptive Beam Search for Provably Accurate Graph-Based Nearest Neighbor Search}

\author{
    Yousef Al-Jazzazi\\ New York University\\ Abu Dhabi\\ \texttt{ya2225@nyu.edu}
	\And 
	Haya Diwan \\ New York University\\ \texttt{hd2371@nyu.edu}
	\And 
 	Jinrui Gou \\ New York University\\ \texttt{jg6226@nyu.edu}
  \And 
	Cameron Musco\\ UMass Amherst\\ \texttt{cmusco@cs.umass.edu}
	\And
	Christopher Musco\\ New York University\\ \texttt{cmusco@nyu.edu}
 	\And 
 	Torsten Suel \\ New York University\\ \texttt{torsten.suel@nyu.edu}
}

\begin{document}

\maketitle

\vspace{-.5em}
\begin{abstract}
Nearest neighbor search is central in machine learning, information retrieval, and databases.
For high-dimensional datasets, graph-based methods such as HNSW, DiskANN, and NSG have become popular thanks to their empirical accuracy and efficiency. These methods construct a directed graph over the dataset and perform beam search on the graph to find nodes close to a given query. While significant work has focused on practical refinements and theoretical understanding of graph-based methods, many questions remain.
We propose a new distance-based termination condition for beam search to replace the commonly used condition based on beam width. We prove that, as long as the search graph is \emph{navigable}, our resulting {Adaptive Beam Search} method is guaranteed to approximately solve the nearest-neighbor problem, establishing a connection between navigability and the performance of graph-based search. We also provide extensive experiments on our new termination condition for both navigable graphs and approximately navigable graphs used in practice, such as HNSW and Vamana graphs. We find that Adaptive Beam Search  outperforms standard beam search over a range of recall values, data sets, graph constructions, and target number of nearest neighbors. It thus provides a simple and practical way to improve the performance of popular methods.
\end{abstract}
\vspace{-.5em}

\section{Introduction}

High-dimensional nearest neighbor search is a basic building block in many areas, including image and video processing \cite{GarciaDebreuveBarlaud:2008,JegouDouzeSchmid:2011}, information retrieval \cite{AumullerBernhardssonFaithfull:2020,PapadopoulosManolopoulos:2005}, and algorithm design \cite{CharikarKNS20,karppa22a}. It is  central to modern machine learning, underlying document and media search based on learned embeddings \cite{Bruch:2024,LuoLakshmanShrivastava:2022,mitra2018introduction}, and most retrieval augmented generation (RAG) systems for large-language models \cite{LewisPerezPiktus:2020,mialon2023augmented}. Nearest neighbor search also plays a role in hard-negative mining \cite{XiongXiongLi:2021}, accelerating transformer architectures \cite{KitaevKaiserLevskaya:2020}, and other applications across machine learning \cite{SpringShrivastava:2017}.  

Formally, in the $k$-nearest neighbor search problem, we are given a set of data points, often machine-learned vector embeddings of documents, images, or other media \cite{DevlinChangLee:2019,FormalPiwowarskiClinchant:2021}. We are also given a distance measure, such as the Euclidean distance, or something more exotic like Chamfer distance \cite{JayaramDhulipalaHadian:2024}. The goal is to pre-process the dataset into a search data structure so that, given any query point $q$, we can efficiently find the $k$ data points closest to $q$ with respect to the distance measure. Doing so exactly is notoriously difficult  in high-dimensions, so applications typically rely on 
approximate nearest neighbor (ANN) methods that attempt to find \emph{most} of the $k$ closest neighbors.

Many different approaches have been proposed for ANN search. Popular methods include locality sensitive hashing (LSH) \cite{AndoniIndyk:2008,AndoniIndykLaarhoven:2015,IndykMotwani:1998,LvJosephsonWang:2007}, inverted file indices (IVF) based on product quantization or clustering \cite{JegouDouzeSchmid:2011,JohnsonDouzeJegou:2021,MujaLowe:2014}, and more \cite{BeygelzimerKakadeLangford:2006,Kleinberg:1997,KushilevitzOstrovskyRabani:1998}. In this work, we focus on \emph{graph-based} ANN methods, which have been extensively studied, and are commonly used thanks to strong empirical performance (graph-based methods have topped leader boards at a number of recent ANN competitions \cite{SimhadriAumullerIngber:2024,SimhadriWilliamsAumuller:2022}).

\smallskip
\noindent\textbf{Graph-Based Nearest Neighbor Search.}
The high-level idea of graph-based methods is simple. We construct an index by building a directed graph, $G$, with one node for each data point. 
Given a query, $q$, we search the index by starting at an arbitrary node and performing a greedy graph traversal, exploring neighbors in that graph that are closest to $q$. A specific choice of graph construction and traversal method comprises a particular ``graph-based'' nearest neighbor search method.

Many algorithms for graph construction have been proposed, including the Hierarchical Navigable Small World (HNSW) approach \cite{MalkovYashunin:2020}, Vamana/DiskANN \cite{KrishnaswamyManoharSimhadri:2024,SubramanyaDevvritKadekodi:2019}, Navigating Spreading-out Graphs (NSG) \cite{FuXiangWang:2019}, and others \cite{ManoharShenBlelloch:2024,Vargas-MunozGoncalvesDias:2019} All of these methods construct a graph $G$ that, for a given node $i$, contain out-edges to nearest neighbors of $i$, as well as ``long range'' connections to nodes far away from $i$. Such constructions are loosely motivated by the concept of \emph{navigability}, which dates back to pioneering work on local graph routing by Kleinberg \cite{Kleinberg:2000b,Kleinberg:2000} and  Milgram \cite{Milgram:1967}. We provide a formal definition of navigability in \Cref{sec:prelim}, but the property roughly guarantees that there is a path from any node $i$ in $G$ to any node $j$ so that distance to $j$ strictly decreases along the path.

While graph constructions vary greatly, the choice of greedy traversal method used in graph-based nearest neighbor search has seen less innovation. A generalization of greedy search called \emph{beam search} is almost ubiquitous. Parameterized by a  beam width $b \ge k$, beam search maintains a list of $b$ candidate nearest neighbors and computes the query's distance to each of those candidates' neighbors, updating them until it fails to find any better candidates. See \Cref{sec:techOverview} for a formal description.

\smallskip
\noindent\textbf{Our Contributions.}
While graph-based ANN methods have seen significant practical success, their performance is poorly understood from a theoretical perspective. This is in contrast to methods like locality sensitive hashing, for which it is possible to prove strong worst-case approximation guarantees \cite{AndoniIndyk:2008,AndoniIndykNguyen:2014}. A lack of theory makes it difficult to iterate on and improve existing graph-based methods, and to understand the limitations of these methods. We aim to address this theory-practice gap, and in turn, introduce principled improvements to existing methods.  

In particular, we re-examine the ubiquitous beam search method, showing that it can be viewed as a specific stopping rule for a much more general search procedure. This perspective motivates a new algorithm called Adaptive Beam Search, which stops searching for candidates based on a \emph{distance-based criterion} instead of a fixed beam width, $b$. Our main theoretical result (\Cref{thm:main}) is to prove that Adaptive Beam Search returns provable approximate nearest neighbors \emph{whenever the search graph $G$ is navigable}. To the best of our knowledge, this result is the first to theoretically connect the performance of greedy search (specifically, beam search) to the property of navigability. 

Moreover, our theoretical results translate into practical performance. We perform an extensive experimental evaluation of Adaptive Beam Search, comparing it to fixed-width beam search over a wide range of data sets, graph constructions, recall values, and target number of nearest neighbors. The method universally outperforms classic beam search, typically providing a $10-50\%$ reduction in the number of distance computations required for a given level of recall. Moreover, Adaptive Beam Search can be implemented with only minor code changes to existing graph-based libraries. We thus hope that, beyond its theoretical relevance, the method will have practical impact. 

\noindent
\textbf{Roadmap.} The remainder of this paper is organized as follows. In \Cref{sec:prelim}, we discuss technical preliminaries and related work. In \Cref{sec:techOverview}, we introduce our Adaptive Beam Search method and its motivating ideas. In \Cref{sec:theory}, we prove that  Adaptive Beam Search solves approximate nearest neighbor search on navigable graphs (\Cref{thm:main}). In \Cref{sec:experiments}, we evaluate \ourmethod{} on sparse navigable graphs and common heuristic graph constructions including HNSW and Vamana.

\section{Background and Related Work}
\label{sec:prelim}
We start by defining notation used throughout. Our goal in this paper is to find nearest neighbors in a metric space $\mathcal{X}$ equipped with a distance function $d: \mathcal{X}\times \mathcal{X}\rightarrow \R^+$.\footnote{In particular, we just require that for all $i,j,k \in \mathcal{X}$, $d(i,j) = d(j,i)$, $d(i,j) > 0$ when $i \neq j$, $d(i,i) = 0$, and $d(i,j) + d(j,k) \geq d(i,k)$ (i.e., triangle inequality holds).} We are given a database of $n$ items in $\mathcal X$, which we label $\{1, \ldots, n\}$. We want to find the nearest $k \leq n$ items to a given query $q \in \mathcal{X}$. E.g., for $k = 1$, the goal is to find $\argmin_{j\in \{1, \ldots, n\}} d(q, j)$. To avoid corner cases, we assume items in the database are unique, i.e., $d(i,j) > 0$ for all $i,j \in \{1, \ldots, n\}$, $i \neq j$. 

In practice, the $n$ database items and the query $q$ are usually associated with vectors (e.g., machine learned embeddings) $\bv{x}_1,\ldots,\bv{x}_n$ and $\bv{x}_q \in \R^m$. The distance function $d(i,j)$ is chosen to be some function of these vectors, e.g.,  the Euclidean distance, $d(i,j) = \|\bv{x}_i - \bv{x}_j\|_2$.

\smallskip

\noindent\textbf{Graph Navigability.} 
Our theoretical guarantees assume use of a \emph{navigable search graph} over $n$ nodes corresponding to our $n$ database items.
While the term ``navigable'' is sometimes used informally in the literature, we use the following precise definition. Consider a directed graph $G = (V,E)$, with $V = \{1, \ldots, n\}$. For a node $x$, let $\mathcal{N}_G(x)$ denote its set of  \emph{out-neighbors}. Define:
\begin{definition}[Navigable Graph]
  \label{def:nav}
  A directed graph $G$ is \emph{navigable} under distance function $d$ if for any nodes $x,y \in \{1, \ldots, n\}$ with $d(x,y) > 0$, there is some $z\in \mathcal{N}_G(x)$ with $d(z,y) < d(x,y)$. 
\end{definition}
Navigability ensures that, for any starting node $s$ and target node $t$, a standard greedy search where we always move to the neighbor of the current node closest to $t$, always converges to $t$. 

When all distances between $\{1, \ldots, n\}$ are unique (this can be ensured by simply tie-breaking based on node id) it was recently shown that \emph{any data set} has an efficiently computable navigable graph with average degree $O(\sqrt{n \log n})$ for \emph{any distance metric} \cite{DiwanGouMusco:2024}. While the above bound is nearly optimal for worst-case data sets, much sparser navigable graphs often exist. For the Euclidean distance in $m$ dimensions, Arya and Mount construct navigable graphs with degree $2^{O(m)}$ \cite{AryaMount:1993}. For general metrics, Indyk and Xu construct navigable graphs with degree $2^{O(m')\log \Delta}$ where $m'$ is the doubling dimension of the data under $d$ and $\Delta = {\max_{i,j}d(i,j)}/{\min_{i,j}d(i,j)}$ is the dynamic range \cite{IndykXu:2023}.

Why do we focus on navigability? Navigability has become a standard notion of ``quality'' for graphs used in nearest neighbor search. Indeed, the term lends its name to popular graph-based search methods such as the Navigable Small World (NSW) \cite{MalkovPonomarenkoLogvinov:2014} and  Hierarchical Navigable Small World (HNSW)  \cite{MalkovYashunin:2020} methods. Neither of these methods  construct graphs that are \emph{provably} navigable, although they produce graphs that should be approximately navigable in practical settings.
Surprisingly, however, to the best of our knowledge, no prior work formally links the accuracy of graph-based search to this intuitive notion of graph quality. As discussed, a major goal here is to address this theory-practice gap, and to use the resulting theory to propose new practical algorithms. 

Related to our approach is a recent paper by Indyk and Xu \cite{IndykXu:2023}, which proves accuracy guarantees for standard beam search under the assumption that the search graph is ``$\alpha$-shortcut reachable'', a strictly stronger criterion than navigability. A graph is $\alpha$-shortcut reachable if, for all $x,y \in \{1, \ldots, n\}$ with $d(x,y) > 0$, there is some $z\in \mathcal{N}_G(x)$ with $\alpha \cdot d(z,y) < d(x,y)$ for some parameter $\alpha \ge 1$. 
Indeed, navigability exactly corresponds to this definition with $\alpha = 1$. However, the result from \cite{IndykXu:2023} only yields a bounded approximation factor for $\alpha > 1$ (concretely,  they obtain approximation factor $\frac{1+\alpha}{1-\alpha}$). Thus, obtaining theoretical results for graphs that are simply navigable remains an open question. 

One reason this question is of practical importance is that navigable graphs can in general be much sparser than $\alpha$-shortcut reachable graphs. While it is possible to construct a navigable graph with average degree $O(\sqrt{n \log n})$ for any database under any metric (under the mild assumption of unique distances) \cite{DiwanGouMusco:2024}, it is not hard to observe that for any fixed $\alpha > 1$, even a random point set in $O(\log n)$-dimensional Euclidean space does not  admit any sparse $\alpha$-shortcut reachable graph (i.e., with average degree $< n -1$) with high probability (see \Cref{sec:alphaFailure} for details). 


\subsection{Additional Related Work.}
Beyond \cite{IndykXu:2023}, a few other papers have studied graph-based ANN search from a theoretical perspective. E.g., \cite{Laarhoven:2018} and \cite{ProkhorenkovaShekhovtsov:2020} study time-space tradeoffs akin to those available for LSH methods, but only for random data. More significant work has focused on practical algorithmic improvements. E.g., work has studied parallel implementations \cite{ManoharShenBlelloch:2024}, methods for dynamic datasets \cite{SinghSubramanyaKrishnaswamy:2021,XuManoharBernstein:2025},  distance approximations \cite{XuSilwalIndyk:2024}, graph pruning \cite{graspWSDM2022}, filtered search \cite{GollapudiKariaSivashankar:2023}, and search with coverage criteria \cite{AnandIndykKrishnaswamy:2025}. We are not aware of work that, like ours, studies significant alternatives to beam width-based termination in beam search, although certain modifications, such as smarter initialization strategies and early stopping criteria have been studied \cite{li:sigmod20,ManoharKimBlelloch:2025,ZhaoTianHuang:2023}.

\section{Adaptive Beam Search}
\label{sec:techOverview}
Beam search is the defacto search method used for graph-based ANN \cite{MalkovYashunin:2020,SubramanyaDevvritKadekodi:2019}. We start with a key observation: beam search can be reframed by decoupling the method into two key components 1) a \emph{search order}, determined by a  method for traversing the search graph to find candidate nearest neighbors and 2) a \emph{stopping criterion}, which governs when the algorithm stops considering candidates. 

Our \ourmethod{} method modifies the standard beam search algorithm only by changing the stopping criterion. The search order remains the same. Surprisingly, even this simple change leads to an algorithm that both enjoys strong theoretical approximation guarantees when the underlying graph is navigable (see \Cref{thm:main}) and outperforms standard beam search empirically.

We suspect the ``decoupled view'' of beam search is not novel, but we have not seen it presented. So, in the next section, we detail this reframing and show how a change in stopping criterion yields other search algorithms, like simple greedy search and \ourmethod{}. We intuitively motivate the stopping criterion used in \ourmethod{} before formally analyzing the method in \Cref{sec:theory}.

\subsection{Decoupling Beam Search as Ordered Traversal With a Stopping Condition}

To be concrete, we provide pseudocode for a generic version of beam search in \Cref{alg:gen_beam_search}. Implementation details are deferred to \Cref{sec:pseudo}. 
 Importantly, such details do not affect the number of \emph{distance computations} performed by the algorithm -- i.e., how many times we evaluate $d(q,i)$ for a query point, $q$, and candidate nearest neighbor, $i$. Distance computations typically dominate the cost of search in practice and, indeed, for the stopping criteria considered in this paper, all other operations can be implemented in time nearly-linear in the number of such computations.  

\begin{algorithm}[h!]
    \caption{Generalized Beam Search}
    \label{alg:gen_beam_search}
    \begin{algorithmic}[1]
        \Require Search graph $G$ over nodes $\{1, \ldots, n\}$, starting node $s$, distance function $d$, query $q$, target number of nearest neighbors $k$.  
        \Ensure Set of $k$ nodes ${\mathcal{B}} \subset \{1, \ldots, n\}$, where each $x\in {\mathcal{B}}$ is ideally close to $q$ with respect to $d$.
        \algrule
        \State Initialize min-priority queues $\mathcal{C}$ and $\mathcal{D}$.\hspace{.5em} 
        \Comment{Elements are nodes, priorities are distances to $q$. $\mathcal{D}$ contains all discovered nodes. $\mathcal{C}$ contains discovered nodes that are not yet expanded.}
        \State Insert $(s, d(q, s))$ into $\mathcal{C}$ and $\mathcal{D}$. 
        \While{$C$ is not empty}
            \State $(x,d(q,x)) \gets \textrm{extractMin}(\mathcal C)$. \Comment{Pop min. distance node.}
            \If{$x$ satisfies \textcolor{red}{\textbf{[termination condition]}}} 
            \State \textbf{break} 
            \EndIf
            \State For all $y \in \mathcal{N}_G(x)$, if $y$ is not in $\mathcal{D}$, insert $(y,d(q,y))$ into $\mathcal{C}$ and $\mathcal{D}$. \Comment{Expand node $x$.}
        \EndWhile
        \State Obtain $\mathcal{B}$ by running extractMin $k$ times on $\mathcal{D}$, which returns the $k$ elements with the smallest distances from the query, $q$.
    \end{algorithmic}
\end{algorithm}

\Cref{alg:gen_beam_search} maintains a queue of ``discovered nodes'' $\mathcal D$ whose distances to $q$ have been computed. It repeatedly ``expands'' the nearest discovered (and not previously expanded) node to $q$ by adding its neighbors to the queue (Line 6). It does so until this nearest node triggers the termination condition in Line 5. The choice of termination condition leads to various versions of greedy search, including beam search and our new distance based \ourmethod{} method. In particular, we have: 
\begin{description}[style=unboxed,leftmargin=0cm]
    \item[Classic Greedy Search.] Terminate if there are at least:
    \begin{align}
    \label{eq:greedy_rule}
        k \text{ items } j_1, \ldots, j_k\in \mathcal{D} \text{ with } d(q, j_i) \leq d(q,x).
    \end{align}
    \item[Beam Search, with beam-width parameter $\mathbf{b \geq k}$.] Terminate if there are at least\footnote{\textit{Remark on implementation:} For beam search, it is easy to see that a node $x$ will always satisfy termination condition \eqref{eq:beam_rule} if it is not one of the closest $b$ neighbors to $q$ in $\mathcal{D}$. So, instead of maintaining two priority queues, it is more computationally efficient to maintain a sorted list of the $b$ closest nodes discovered so far. This is what is done in typical implementations of beam search \cite{SubramanyaDevvritKadekodi:2019}, and in our more detailed pseudocode in \Cref{sec:pseudo}.} :
    \begin{align}
    \label{eq:beam_rule}
        b \text{ items } j_1, \ldots, j_b\in \mathcal{D} \text{ with } d(q, j_i) \leq d(q,x).
    \end{align}
    \item[\ourmethod{} (our method) w/ parameter $\boldsymbol{\gamma}$.] Terminate if there are at least:
        \begin{align}
    \label{eq:our_rule}
        k \text{ items } j_1, \ldots, j_k\in \mathcal{D} \text{ with } (1+\gamma)  \cdot d(q, j_i) \leq d(q,x).
    \end{align}
\end{description}

The rule for greedy search is simple: we terminate if we have already found $k$ points closer to $q$ than the current candidate considered for expansion. For $k = 1$, it takes a moment to confirm that this criterion yields a method that is exactly equivalent to the more typical way of presenting greedy search: starting at $s$, move to the neighboring node nearest to $q$, terminating if there is no neighbor closer than the current node. 
For $k=1$, greedy search is known to converge to the exact nearest neighbor if there is some $x\in \{1, \ldots, n\}$ for which $d(x,q) = 0$ \emph{and} the search graph $G$ is navigable \cite{DiwanGouMusco:2024,Kleinberg:2000b,Milgram:1967}. However, no comparable guarantees hold for $k > 1$ or when $q$'s nearest neighbor is not at distance $0$, which is typical in practice. Moreover, greedy search performs poorly empirically, easily getting stuck in local minima  and failing to find good approximate nearest neighbors.

\subsection{Relaxing Greedy Search}\label{sec:relaxing}

The goal of beam search is to avoid such accuracy issues. It does so by relaxing the stopping criterion from greedy search: in particular, by \eqref{eq:beam_rule}, we only terminate if we have found $b \geq k$ nodes closer to the query $q$ than our current node $x$. When $b = k$, the algorithms are identical. When $b > k$, greedy search explores a prefix of the nodes explored by beam search, which simply terminates the search at a later point. Beam search is thus guaranteed to obtain a more accurate result than greedy search, at the cost of an increased number of distance computations.

With the above view in mind, many other relaxations of the greedy search termination condition given in \eqref{eq:greedy_rule} become apparent. In \eqref{eq:our_rule}, we introduce a slack parameter $\gamma \geq 0$ and only terminate if $x$ is further from $q$ than the $k^\text{th}$ best discovered point by a factor of $1+\gamma$. Setting $\gamma = 0$ recovers greedy search, and larger values of $\gamma$ will cause the search process to terminate later, yielding a better result, but at the cost of a higher runtime. This simple idea yields our \ourmethod{} procedure. 

While intuitively similar to beam search, a key difference of this ``distance based'' criteria is that it naturally adapts to the \emph{query difficulty}. For simplicity, consider the case of $k=1$. Greedy search tends to perform worse when there are many ``false nearest neighbors'' in a dataset. For example, suppose there is just one best point $x^* \in \{1, \ldots, n\}$ with $d(q,x^*) = 1$, but $m$ other points $y_1, \ldots, y_{m}$ with $d(q,y_i) = 1.01$. Unless we choose beam width $b  = \Omega(m)$, it is likely that more than $b$ points at distance $1.01$ will get added to $\mathcal D$, causing the search to terminate before finding $x^*$. 
In contrast, as long as $\gamma > .01$, \ourmethod{} will continue to search through all of the $y_i$ points before terminating. In contrast, Adaptive Beam Search will more quickly terminate search if it becomes apparent that all remaining candidates are too far away to be useful in finding addition nearest neighbors. Indeed, a criterion identical to \ourmethod{} has been suggested as an ``early stopping'' heuristic in work on practical graph-based nearest neighbor search \cite{ManoharShenBlelloch:2024,ManoharKimBlelloch:2025}.

\begin{figure*}[b]
\vspace{-.5em}
    \centering
        \begin{subfigure}[b]{0.325\linewidth}
        \centering
        \includegraphics[width=\linewidth]{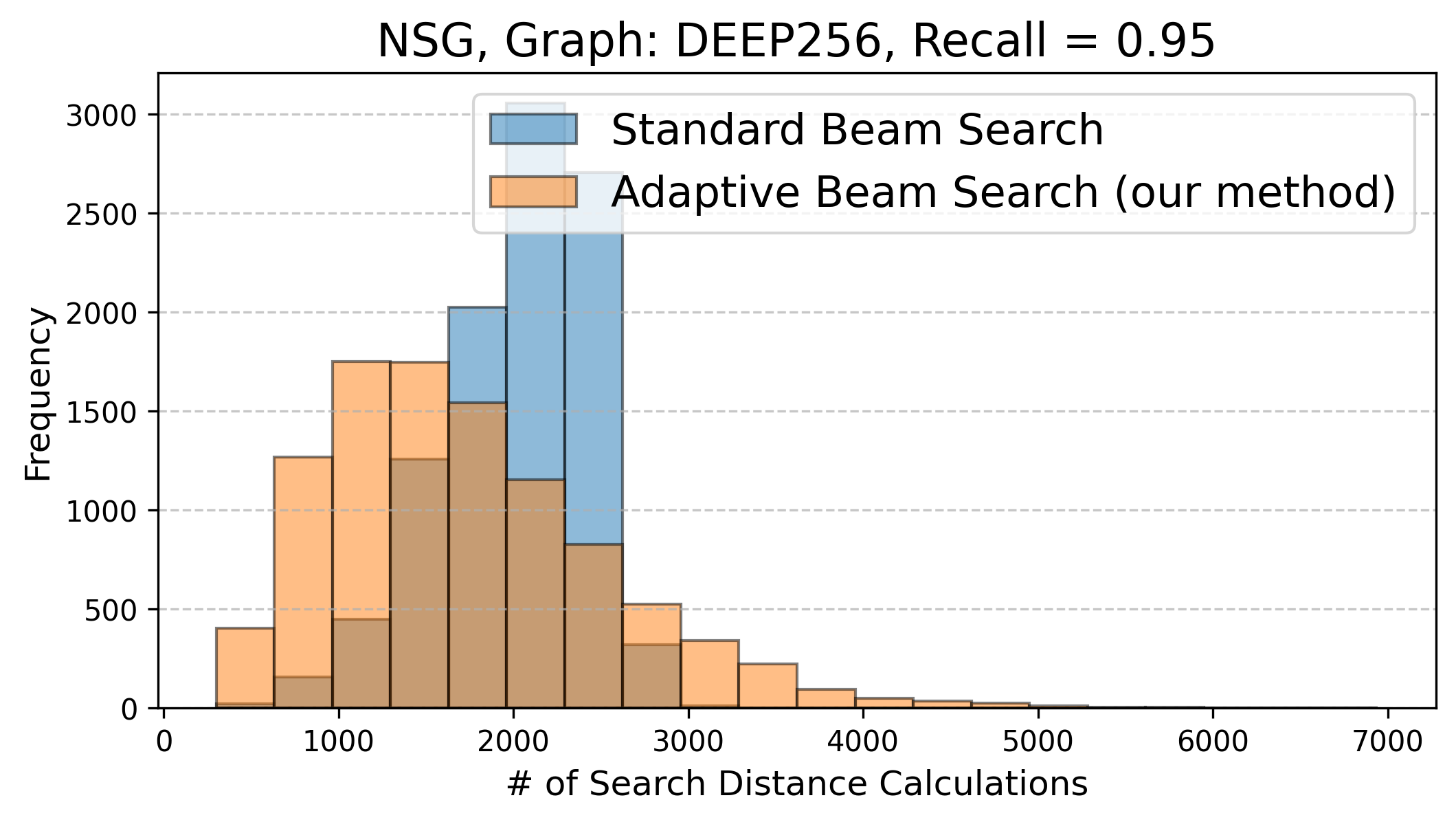}
        \label{fig:navig_histogram_1}
    \end{subfigure}
        \hfill
    \begin{subfigure}[b]{0.325\linewidth}
        \centering
        \includegraphics[width=\linewidth]{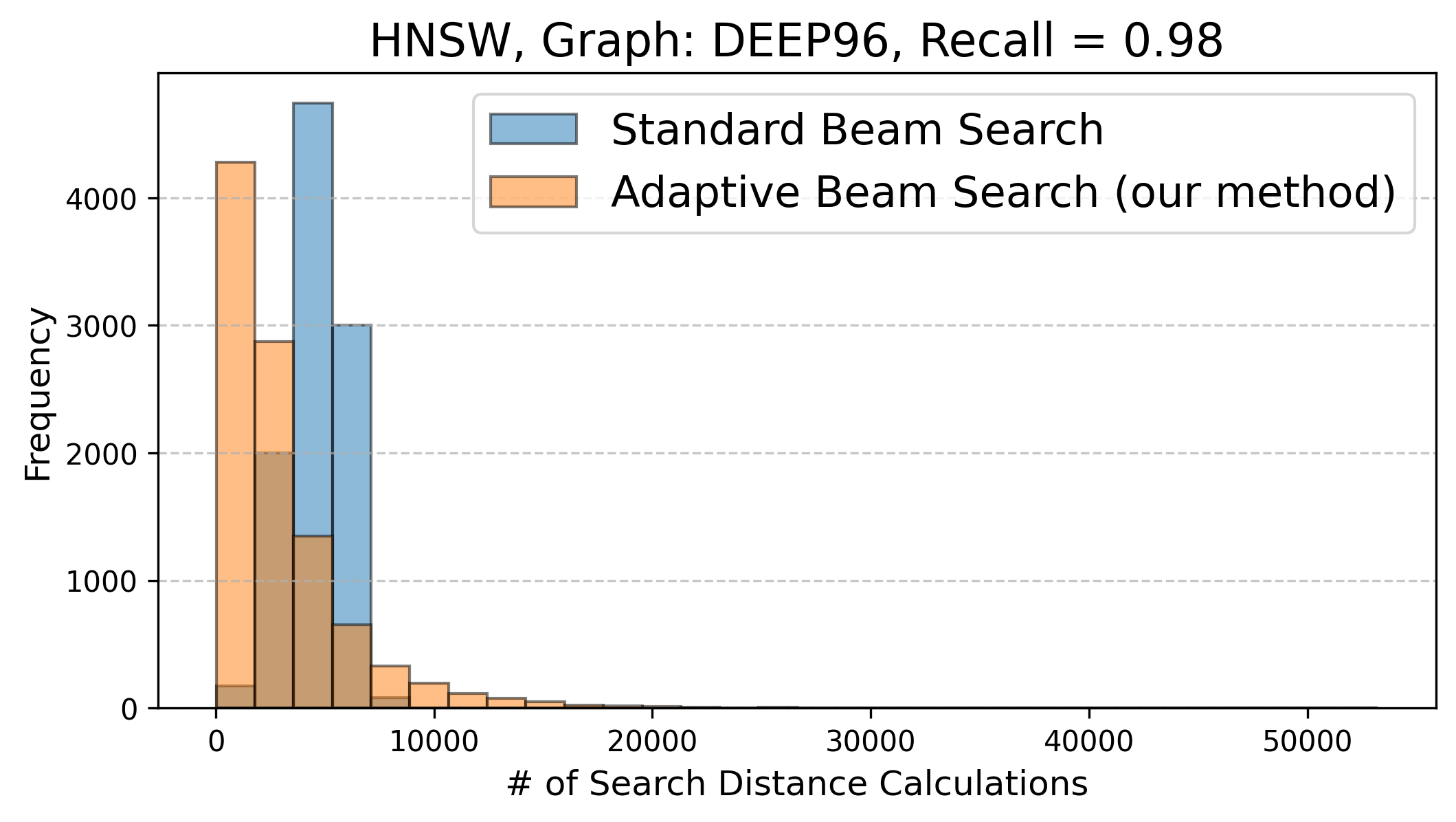}
        \label{fig:navig_histogram_1}
    \end{subfigure}
    \hfill
    \begin{subfigure}[b]{0.325\linewidth}
        \centering
        \includegraphics[width=\linewidth]{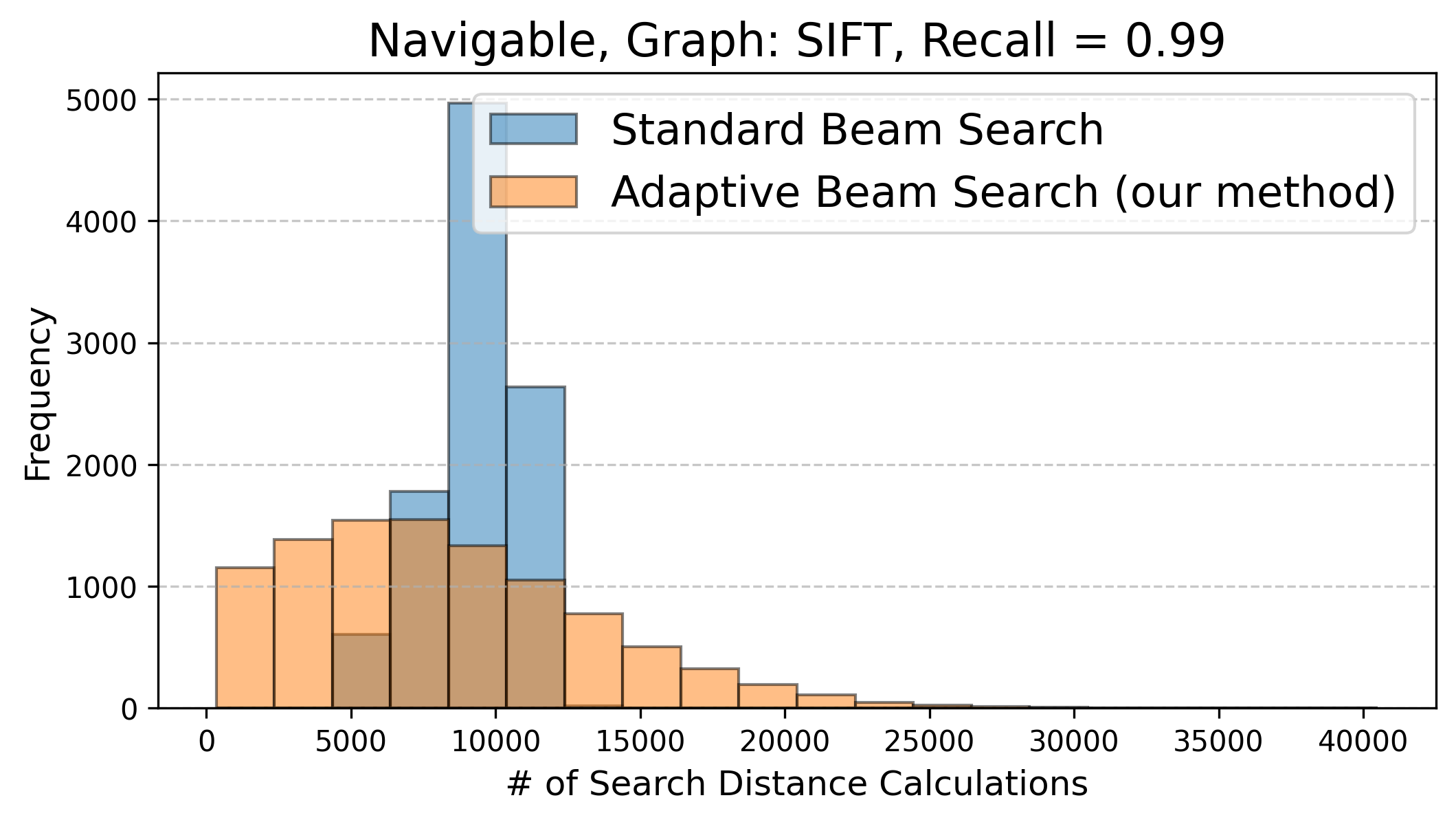}
        \label{fig:navig_histogram_bw}
    \end{subfigure}
    \vspace{-.5em}
    \caption{Histograms for the number of distance computations performed by standard beam search and our \ourmethod{} method when answering 10,000 queries for various datasets and search graphs (see \Cref{sec:experiments} for details). For a fair comparison, the $b$ parameter in beam search and $\gamma$ parameter in \ourmethod{} were tuned to achieve a fixed level of recall for the batch of queries. The histograms for \ourmethod{} are consistently flatter, confirming the intuition that it better adapts to query difficulty, leading to fewer distance computations on average.}
    \label{fig:intro_teaser}
\end{figure*}

The intuition that \ourmethod{} adapts to query hardness shows clearly in our experiments: as seen in \Cref{fig:intro_teaser}, the distribution of distance computations used by \ourmethod{} varies more widely, as fewer computations are used for ``easier'' queries.
As a result, across a variety of data sets and search graphs,  \ourmethod{} consistently outperforms classic beam search in terms of total  distance computations required to achieve a certain level of recall for a given query set.

\section{Theoretical Analysis}\label{sec:theory}

We support the improved empirical performance of \ourmethod{} with strong theoretically guarantees. Formally, we prove that the method is \emph{guaranteed} to solve the approximate nearest neighbor search problem, assuming that the search graph $G$ is navigable (\Cref{def:nav}):
\begin{theorem}
\label{thm:main}
Suppose $d$ is a metric on $\mathcal{X}$ and $G$ is navigable under $d$. Then for any query $q \in \mathcal{X}$, if \ourmethod{} -- i.e., \Cref{alg:gen_beam_search} with stopping criterion \eqref{eq:our_rule} -- is run with parameter $0 < \gamma \leq 2$, it is guaranteed to return a set of $k$ points $\mathcal{B}$ such that:
\begin{align*}
&\text{for all } v \in \{1,\ldots, n\} \setminus \mathcal{B}, & d(q, v) &\geq \frac{\gamma}{2} \max_{j\in \mathcal{B}} d(q, j).
\end{align*}
\end{theorem}
Notably, setting $\gamma = 2$, we ensure that all points not returned by the algorithm are at least as far from $q$ as \emph{every point} in $\mathcal{B}$. Thus, \textbf{for $\gamma = 2$, \ourmethod{} on a navigable graph is guaranteed to \emph{exactly} solve the $k$-nearest neighbor problem.} For smaller  $\gamma$, the method obtains an approximate solution: no point in $\mathcal{B}$ can be further from $q$ than any point not returned by more than a $2/\gamma$ factor.\footnote{Many existing theoretical guarantees for approximate nearest neighbor search, such as those for LSH and related methods \cite{AndoniIndyk:2008,AndoniIndykNguyen:2014,IndykMotwani:1998} focus on the case of $k = 1$. A rephrasing of our result in this case is that \ourmethod{} returns an approximate nearest neighbor $\tilde{x}$ with $
d(q,\tilde{x}) \leq \frac{2}{\gamma} \cdot  \min_{x\in \{1, \ldots, n\}}d(q,x).$}

 We can see that
\Cref{thm:main} proves a trade-off between runtime and accuracy: smaller values of $\gamma$ lead to a strictly faster algorithm (since termination is earlier) but a worse approximation guarantee. 
While our result falls short of proving worst-case runtime guarantees, to the best of our knowledge, it is the first result linking the accuracy of a natural greedy search method to the notion of graph navigability.
Importantly we note that, unlike our Adaptive Beam Search, a result like \Cref{thm:main} \emph{cannot be proven} for standard beam search. In particular, in \Cref{sec:beamFailure} we prove:

\begin{claim}
    Standard beam search with beam width $b \le n-3$ fails to approximately solve the nearest neighbor search problem on navigable graphs for any finite approximation factor. 
    \end{claim}
    
    Concretely, for any finite $C$, we can construct a set of $n$ points in $2$-dimensional Euclidean space and a navigable graph $G$ such that, for some query point $q$, beam search run on $G$ with beam width $b \le n-3$ returns $\tilde{x}$ with $\displaystyle d(q,\tilde{x}) \geq C\cdot  \min_{x\in \{1, \ldots, n\}}d(q,x)$. 
\begin{proof}[Proof of \Cref{thm:main}]
  Our proof will use the terms ``discovered'' and ``expanded'' to identify nodes in $\{1, \ldots, n\}$. We consider a node $j$ ``discovered'' if $j\in \mathcal{D}$ when \Cref{alg:gen_beam_search} terminates; i.e., we have evaluated the distance between $j$ and $q$. We consider a node $j$ ``expanded'' if $j$ is discovered and, at some point, was both popped off $\mathcal{C}$ on Line 4 \emph{and} did not cause the termination condition on Line 5 to be triggered. This ensures that all of its out-neighbors are discovered (see Line 7). 
  
  Note that all discovered nodes are added to both $\mathcal D$ and $\mathcal{C}$. Formally, if the algorithm terminates because the condition is true for some $x_{term}$, then $\mathcal{C}\cup \{x_{term}\}$ is the set of discovered but not yet expanded nodes, so the set of expanded nodes is $\mathcal{D}\setminus (\mathcal{C}\cup \{x_{term}\})$. 

 Let $\mathcal{B}$ be the set of nodes returned upon termination and let $\tilde{x} = \argmax_{x\in \mathcal{B}} d(q,x)$ be the $k^\text{th}$ furthest point from $q$ in that set. Since $G$ is navigable, and since we assume data points are unique, there must be a path in $G$ from any node $x$ to any other node $y$ (consisting of nodes that get monotonically closer to $y$); i.e., $G$ is strongly connected. Thus, if \Cref{alg:gen_beam_search} terminates because an empty queue $\mathcal C$ causes the while loop to terminate, then all nodes in the graph must have been discovered, and so $\mathcal B$ contains the exact $k$ nearest neighbors to $q$, and the theorem holds immediately.

Thus, it suffices to consider the case when termination occurs because some node $x_{term}$ causes the termination condition in Line 5 to evaluate to true and the while loop to break early. We first claim:
\begin{claim}\label{clm:expanded}
When \Cref{alg:gen_beam_search} terminates, $\tilde{x}$ is guaranteed to have been expanded.
\end{claim}
  To see that this claim holds note that, by termination condition \eqref{eq:our_rule}, it must be that $d(q,x_{term}) \geq (1+\gamma) d(q,\tilde{x})$ and thus $d(q,x_{term}) > d(q,\tilde{x})$.\footnote{The strict inequality clearly holds when $d(q,\tilde{x}) > 0$ since $\gamma > 0$. When $d(q,\tilde{x}) = 0$ it holds because database items are assumed to be unique, so we cannot also have $d(q,{x}_{term}) = 0$.} I.e., $\tilde{x}$ is closer to $q$ then $x_{term}$. Thus, $\tilde{x}$ must have already been popped off $\mathcal{C}$ and expanded before $x_{term}$ was popped off $\mathcal{C}$.

With \Cref{clm:expanded} in place, we can get into our main proof. Our goal is to prove that for all $z\notin \mathcal{B}$,
\begin{align}
  \label{eq:to_prove_k1}
d(q,z) \geq  \frac{\gamma}{2} d(q,\tilde{x}). 
\end{align}
It suffices to prove the claim for all undiscovered nodes $z\notin \mathcal{D}$, since if $z\in \mathcal{D}$ and $d(q,z) < \frac{\gamma}{2} d(q,\tilde{x})$, then $z$ is closer to $q$ than $\tilde{x}$ and would have clearly been included in $\mathcal{B}$ (recall that $\gamma \leq 2$). 

Now, suppose by way of contradiction that \eqref{eq:to_prove_k1} is not true, i.e., that there is some undiscovered node $z\notin \mathcal{D}$ with $d(q,z) <  \frac{\gamma}{2} d(q,\tilde{x})$. We first observe that such a $z$ cannot be an out neighbor of $\tilde{x}$: since $\tilde{x}$ is expanded by \Cref{clm:expanded}, all of its neighbors are discovered, i.e., all are in $\mathcal{D}$. 

Since $G$ is navigable and all database items are unique, there must be some directed path $\mathcal{P}$ from $\tilde{x}$ to $z$ consisting of points that get monotonically closer to $z$. Moreover, since $z\notin \mathcal{N}_G(\tilde{x})$, $\mathcal{P}$ must have length $\ell \geq 2$. Denote the elements of $\mathcal{P}$ by $
  \mathcal{P} = \{\tilde{x} = p_0 \rightarrow p_1 \rightarrow \ldots \rightarrow p_\ell = z\}.$ We have for all $1 \leq i \leq \ell$, $d(z, p_{i-1}) > d(z, p_{i})$. 
We make the following claim:
\begin{claim}\label{clm:w} For any $z \notin \mathcal D$, there exists some node $w \in \{p_1, \ldots, p_{\ell-1}\}$ along the path from $\tilde x$ to $z$ that has been {discovered} but not {expanded}.
\end{claim} 
\begin{proof}
First observe that $p_1$ must be discovered since, by  \Cref{clm:expanded}, $\tilde{x}$ was expanded and $p_1$ is an out-neighbor of $\tilde{x}$. Furthermore, if $p_{i-1}$ is discovered \emph{and} expanded then $p_{i}$ must be discovered. So, inductively we see that there are two possible cases: either there is some $i < \ell$ for which $p_{i}$ is discovered but not expanded (as desired) or $p_i$ is discovered \emph{and} expanded for all $i < \ell$. However, the second case is impossible since $z$ is not in $\mathcal{D}$ and it would be if $p_{\ell-1}$ was expanded. We conclude the claim that there is some $w \in \{p_1, \ldots, p_{\ell-1}\}$ that is {discovered} but not {expanded}. 
\end{proof}

Consider the unexpanded node $w$ guaranteed to exist by \Cref{clm:w}. 
When the algorithm terminates, it must be that: 
\begin{align}
  \label{eq:c_condition}
d(q, w) \geq (1+\gamma) d(q, \tilde{x}).
\end{align}
If $w = x_{term}$ this is trivially true as a consequence of the termination rule \eqref{eq:our_rule}. Otherwise, if \eqref{eq:c_condition} were not true, then $w$ would be closer to $q$ than $x_{term}$ and it would have been popped off $\mathcal{C}$ before $x_{term}$ and expanded. With \eqref{eq:c_condition} in place, we are ready to obtain our contradiction.
By triangle inequality (since $d$ is a metric) and our supposition that $d(q,z) <  \frac{\gamma}{2} d(q,\tilde{x})$, we have:
\begin{align*}
 d(\tilde{x},z) \leq  d(\tilde{x},q) + d(q,z) <  \left(1+ \frac{\gamma}{2}\right)d(q,\tilde{x}).
\end{align*}
Combined with another application of triangle inequality and the fact the $d(w,z) < d(\tilde{x},z)$, we have
\begin{align*}
  d(w,q) \leq d(w,z) + d(z,q) &< d(\tilde{x},z) + d(z,q) 
  < \left(1+ \frac{\gamma}{2}\right)d(q,\tilde{x}) + \frac{\gamma}{2}d(q,\tilde{x})  = (1+\gamma)d(q,\tilde{x}).
\end{align*}
However, this claim contradicts \eqref{eq:c_condition}. Thus, there cannot exist any $z\notin \mathcal{D}$ with $d(q,z) <  \frac{\gamma}{2} d(q,\tilde{x})$. I.e., \eqref{eq:to_prove_k1} holds, proving \Cref{thm:main}. For a geometric illustration of the above proof, see Fig. \ref{fig:proof example}.
\end{proof}
\begin{figure}[hbtp!]
\floatbox[{\capbeside\thisfloatsetup{capbesideposition={left,center},capbesidewidth=6.8cm}}]{figure}[\FBwidth]
{\caption{Visualization of the proof of \Cref{thm:main}. We let $\tilde{d}$ denote $d(q,\tilde{x})$. Our goal is to show that there is no undiscovered $z$ in a ball of radius $\frac{\gamma}{2}\tilde{d}$ around $q$, which is shown with a dotted line. 
If there was, we obtain a contradiction. In particular, if $G$ is navigable, we argued that there must be some unexpanded node $w$ on a path of decreasing distance from $\tilde{x}$ to $z$. Since $w$ is closer to $z$ than $\tilde{x}$, it must lie in a ball of radius $\left(1+\frac{\gamma}{2}\right)\tilde{d}$ around $z$, which is contained in a ball of radius $(1+\gamma)$ around $q$. However, by \eqref{eq:c_condition}, no unexpanded node can lie in that ball.}\label{fig:proof example}}
{\hspace{2em}\begin{tikzpicture}[scale=0.65]
    \def\rSmall{2}
    \def\rLarge{3}
    \def\rInner{1.5}
    \def\rIntermediate{2.5} 
    \def\rOuter{4.5} 

    \def\shiftY{0.7}

    \draw[densely dotted] (0,\shiftY) circle (\rInner);
    \draw (0,\shiftY) circle (\rSmall);

    \coordinate (q) at (0,\shiftY); 
    \coordinate (xt) at (0,\shiftY - \rSmall); 
    \coordinate (z) at (135:\rInner*0.9); 
    \coordinate (w) at (150:\rIntermediate); 

    \draw (z) circle (\rLarge);

    \draw (q) circle (\rOuter);

    \fill (q) circle (1pt) node[above left, yshift=-5pt, xshift=2pt] {$q$};
    \fill (xt) circle (1pt) node[below] {$\tilde{x}$};
    \fill (z) circle (1pt) node[above left, yshift=-4pt, xshift=2pt] {$z$};
    \fill (w) circle (1pt) node[above left, yshift=-4pt, xshift=2pt] {$w$};

    \draw (q) -- node[sloped, above] {$\frac{\gamma}{2} \tilde{d}$} (\rInner, \shiftY);

    \draw (z) -- node[sloped, above, pos=0.5] {$\left(1+\frac{\gamma}{2}\right)\tilde{d}$} +(45:\rLarge);

    \draw[densely dotted] (z) -- (w);
    \draw[densely dotted] (w) -- (xt);
    \draw (xt) -- node[sloped, above] {$\tilde{d}$} (q);

    \draw (q) -- node[sloped, above, pos=0.75] {$(1+\gamma)\tilde{d}$} ++(-45:\rOuter);
\end{tikzpicture}
}
\end{figure}
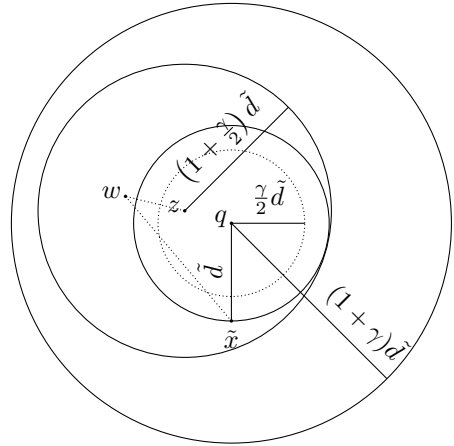

\section{Experiments}\label{sec:experiments}
We now experimentally compare our \ourmethod{} method with standard beam search, demonstrating improved tradeoffs between efficiency and accuracy in a variety of settings. 

\subsection{Experimental Setup} 
\begin{table}[b]
    \caption{Datasets used for evaluation. For further details, refer to \Cref{sec:expDetails}.}
    \label{tab:datasets}
    \smallskip
    \centering
    {
    \begin{tabular}{lcccc}
        \toprule
        \textbf{Dataset} & \textbf{Dimensions} & \textbf{\# of Nodes} & \makecell{ \textbf{\# of Nodes in Navigable} \\ \textbf{Graph Experiments}} & \textbf{\# Query Points} \\
        \midrule
        MNIST \cite{lecun1998gradient}  & 784  & 60K  & 50K & 10K  \\
        SIFT1M \cite{JegouDouzeSchmid:2011} & 128  & 1M  & 100K & 10K \\    
        DEEP96
        \cite{babenko2016efficient}  & 96   & 1M & 100K  &10K\\
        DEEP256 \cite{LiYanZhang:2018} & 256  & 1M  & - & 10K  \\
        GloVe \cite{pennington2014GloVe}  & 200  & 1M & - & 10K  \\
         GIST \cite{oliva2001gist}  & 960  & 1M  & - & 10K  \\

        \bottomrule
    \end{tabular}
    }
\end{table}

\noindent\textbf{Beam Search Algorithms.}
We primarily compare standard beam search (termination condition \eqref{eq:beam_rule}) with \ourmethod{} (termination condition \eqref{eq:our_rule}). To implement \Cref{alg:gen_beam_search} with these termination conditions, we follow the pseudocode in \Cref{sec:pseudo}. For  some settings, we test a third method called \emph{\ourmethod{} V2}, which terminates on node $x$ if
\begin{align}\label{eq:v2}
d(q,x) \ge d_1 + \gamma \cdot d_k,
\end{align}
where $d_1$ and $d_k$ are the distances from the query $q$ to the closest and $k^\text{th}$ closest discovered nodes, respectively. Compared to  \eqref{eq:our_rule},  \eqref{eq:v2} replaces  the threshold $(1+\gamma) \cdot d_k$ with the smaller threshold $d_1 + \gamma \cdot d_k$, leading to more aggressive stopping. Surprisingly, while \eqref{eq:v2} is not a relaxation of greedy search (when $\gamma < 1$, it may stop earlier than greedy search), one can check that \Cref{thm:main} still holds under this condition. This motivates its inclusion in our experiments. However, we observe that \ourmethod{} V2 generally underperforms \ourmethod{}. We leave open developing other stopping conditions that satisfy  bounds similar to \Cref{thm:main} while obtaining strong empirical performance like \ourmethod{} -- see \Cref{sec:hybrid_rule} for some initial explorations.

\noindent\textbf{Comparison Across Recall Values.}
The algorithms discussed above can all trade off accuracy for runtime by adjusting the beam width, $b$, or the parameter $\gamma$. We thus vary these parameters to obtain a range of \emph{recall values}, i.e., the average fraction of the $k$ nearest neighbors found over all queries on a given dataset. Recall is a standard metric for evaluating ANN methods \cite{MalkovYashunin:2020,SubramanyaDevvritKadekodi:2019}.  We compare the methods by plotting the average number of distance computations performed per query to achieve a certain recall value. Since all three methods have essentially identical implementations, running time scales very similarly with the number of distance computations. See \Cref{sec:pseudo} for more details.

\noindent\textbf{Datasets and Graph Constructions.}
We evaluate our \ourmethod{} on six standard benchmark datasets for nearest neighbor search, which are listed in \Cref{tab:datasets}. All datasets consist of real-valued vectors in varying dimensions, and we use Euclidean distance for search. We perform evaluations using a variety of popular heuristic ``approximately navigable'' graphs, along with truly navigable graphs for which the bound of \Cref{thm:main} holds. Specifically, for the heuristic graphs, we use four standard methods: HNSW \cite{MalkovYashunin:2020},  Vamana \cite{SubramanyaDevvritKadekodi:2019}, NSG \cite{FuXiangWang:2019}, and EFANNA \cite{fu2016EFANNA}. Details on how parameters are set for these algorithms are in \Cref{sec:expDetails}. 

To construct the truly navigable graphs, we use the approach of \cite{DiwanGouMusco:2024} to create an initial navigable graph with average degree $O(\sqrt{n \log n})$, and then further prune this graph while maintaining navigability. See \Cref{sec:prune} for details. Pruning reduces the memory footprint of the graph, and results in levels of sparsity closer to that of the heuristic constructions. However, since it is computationally expensive, we only run our navigable graph experiments for random subsets of three of the datasets, with subsample sizes listed in \Cref{tab:datasets}. We believe that our subsample sizes are large enough to be representative. However, it would be interesting to improve the running time of constructing very sparse and truly navigable graphs, so that such graphs can be evaluated for larger datasets. 

\begin{figure*}[t]  
    \centering
    \includegraphics[width=\linewidth]{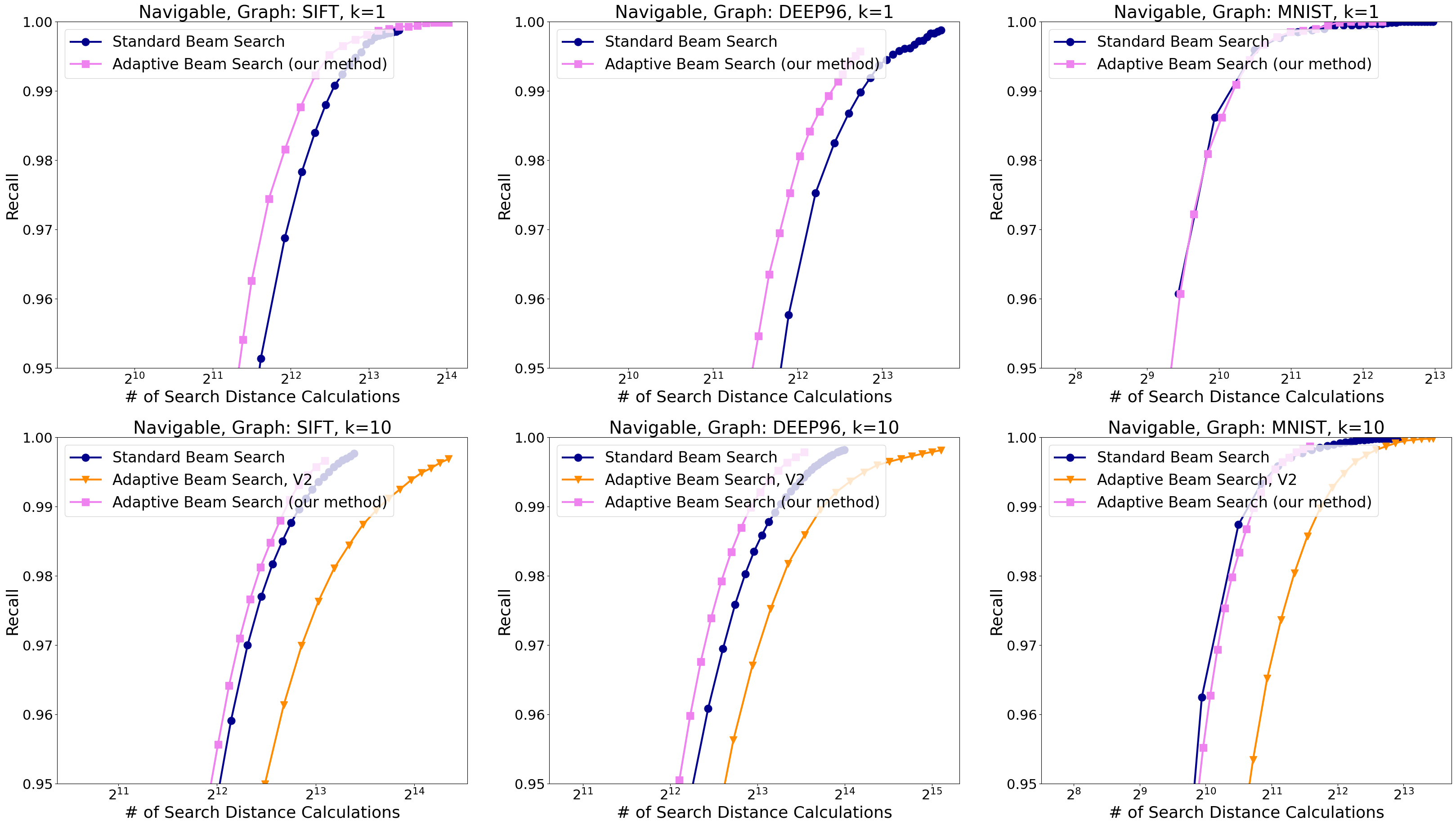}
    \caption{{\bf Navigable Graphs:} Comparison of generalized beam search termination conditions on navigable graphs across three datasets: SIFT1M, DEEP96, and MNIST (columns), with \( k = 1 \), and \( k = 10 \) (rows). \ourmethod{} consistently outperforms standard beam search, while the alternative \ourmethod{} V2 underperforms both by a significant margin. Note that for $k = 1$,  \ourmethod{} and  \ourmethod{} V2 are identical, so only one line is shown.}
    \label{fig:Navigable Graphs}
\end{figure*}

 \subsection{Results}

We now discuss our experimental results on both truly navigable graphs and the commonly used heuristic graphs discussed above. 

\noindent \textbf{Results for Navigable Graphs.}
Results for navigable graphs are shown in Figure \ref{fig:Navigable Graphs} for SIFT, DEEP256, and MNIST for $k = 1$ and $10$. Results for $k = 100$ are included in \Cref{sec:addNav}. The y-axis shows recall, while the x-axis shows the average number of distance calculations per query. \ourmethod{} always performs at least on par with classic beam search, and often significantly better, with up to 30-40\% decrease in distance computations for a given recall.  \ourmethod{} V2 performs worse, so is not evaluated in future experiments.  The underperformance  of \ourmethod{} V2 is further explored in \Cref{sec:histogram}. In a nutshell, when $d_1 \ll d_k$, for small $\gamma$ we might stop when $d(q,x) < d_k$, which means we do not even explore all the neighbors of our current top-$k$ results. If we increase $\gamma$ to avoid this, we terminate too late when $d_1$ is close to $d_k$.

\noindent \textbf{Results for Heuristic Graphs.}
Our results for heuristic graphs with $k = 10$ across three datasets are shown in \Cref{fig:K=10_diff_datasets}. For additional results covering the remaining datasets and values of $k$, see \Cref{sec:addHeuristic}. In all cases, we see that \ourmethod{}  outperforms standard beam search, sometimes marginally, but sometimes by more than a factor of $2$, e.g., on  MNIST. The performance gains are robust to changing the graph construction, indicating that \ourmethod{} is a strong candidate for a drop-in replacement for standard beam search in graph-based ANN.

\noindent\textbf{Adaptivity Across Queries.} As discussed in \Cref{sec:relaxing}, \ourmethod{} seems to outperform standard greedy search because the distance-based stopping criterion is more ``adaptive'' to query difficulty. For hard queries with many approximate nearest neighbors, it tends to use more distance computations. However, the method terminates quickly on easy queries when there are few points with $d(q,x) \le (1+\gamma) d_k$. This phenomenon is illustrated for a sample of settings in  \Cref{fig:intro_teaser}.

\section*{Acknowledgements}
We would like to thank Ramon Li for contributions at an early stage of this work. Christopher Musco was partially supported by NSF Award 2106888. 
\vspace{20em}

\begin{figure*}[hbtp!]  
    \centering
    \includegraphics[width=.98\linewidth]{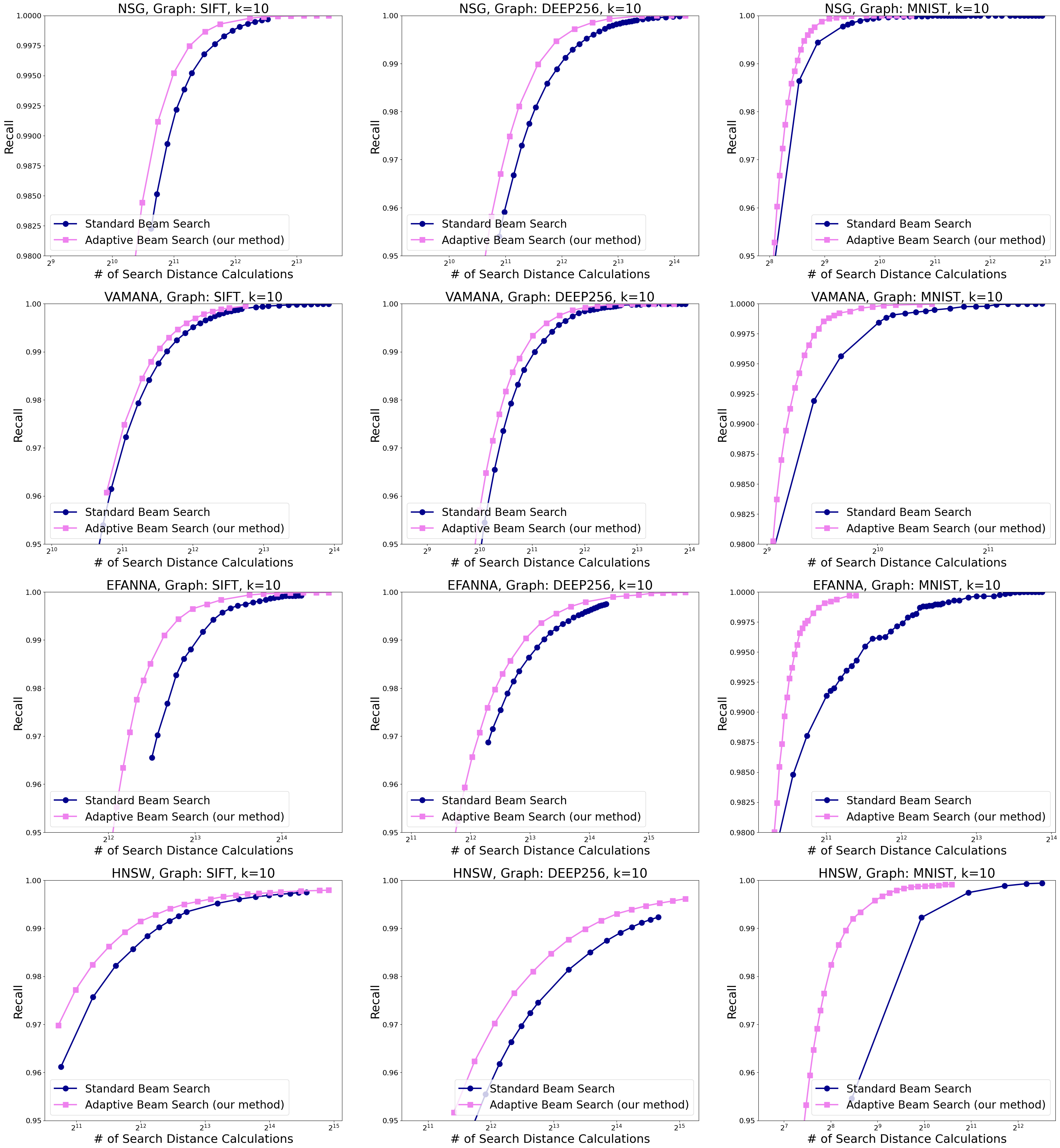}
    \caption{{\bf Heuristic Graphs:} Comparison of generalized beam search termination methods on  heuristic graphs produced by NSG, Vamana, EFANNA, and HNSW (rows), for \( k = 10 \) with 3 datasets: SIFT1M, DEEP256, and MNIST (columns). Adaptive beam search consistently outperforms standard beam search across all cases, sometimes by a significant margin.}
    \label{fig:K=10_diff_datasets}
\end{figure*}

\bibliographystyle{plain}
\bibliography{sample-base}

\begin{thebibliography}{10}

\bibitem{AnandIndykKrishnaswamy:2025}
Piyush Anand, Piotr Indyk, Ravishankar Krishnaswamy, Sepideh Mahabadi, Vikas~C.
  Raykar, Kirankumar Shiragur, and Haike Xu.
\newblock Graph-based algorithms for diverse similarity search.
\newblock {\em \arXiv{2502.13336}}, 2025.

\bibitem{AndoniIndyk:2008}
Alexandr Andoni and Piotr Indyk.
\newblock Near-optimal hashing algorithms for approximate nearest neighbor in
  high dimensions.
\newblock {\em Commun. ACM}, 51(1):117--122, 2008.

\bibitem{AndoniIndykLaarhoven:2015}
Alexandr Andoni, Piotr Indyk, Thijs Laarhoven, Ilya Razenshteyn, and Ludwig
  Schmidt.
\newblock Practical and optimal {LSH} for angular distance.
\newblock In {\em \NIPS{2015}}, 2015.

\bibitem{AndoniIndykNguyen:2014}
Alexandr Andoni, Piotr Indyk, Huy~L. Nguyen, and Ilya Razenshteyn.
\newblock Beyond locality-sensitive hashing.
\newblock In {\em \SODA{2014}}, 2014.

\bibitem{AryaMount:1993}
Sunil Arya and David~M. Mount.
\newblock Approximate nearest neighbor queries in fixed dimensions.
\newblock In {\em \SODA{1993}}, 1993.

\bibitem{AumullerBernhardssonFaithfull:2020}
Martin Aum{\"u}ller, Erik Bernhardsson, and Alexander Faithfull.
\newblock {ANN-Benchmarks}: A benchmarking tool for approximate nearest
  neighbor algorithms.
\newblock {\em Information Systems}, 87, 2020.
\newblock Data accessed at \url{https://github.com/erikbern/ann-benchmarks}.

\bibitem{babenko2016efficient}
Artem Babenko and Victor Lempitsky.
\newblock Efficient indexing of billion-scale datasets of deep descriptors.
\newblock In {\em Proceedings of the IEEE Conference on Computer Vision and
  Pattern Recognition}, pages 2055--2063, 2016.
\newblock Data accessed at \url{https://github.com/erikbern/ann-benchmarks}.

\bibitem{BeygelzimerKakadeLangford:2006}
Alina Beygelzimer, Sham Kakade, and John Langford.
\newblock Cover trees for nearest neighbor.
\newblock In {\em \ICML{2006}}, 2006.

\bibitem{Bruch:2024}
Sebastian Bruch.
\newblock {\em Foundations of Vector Retrieval}.
\newblock Springer, 2024.

\bibitem{CharikarKNS20}
Moses Charikar, Michael Kapralov, Navid Nouri, and Paris Siminelakis.
\newblock Kernel density estimation through density constrained near neighbor
  search.
\newblock In {\em \FOCS{2022}}, pages 172--183, 2020.

\bibitem{DevlinChangLee:2019}
Jacob Devlin, Ming-Wei Chang, Kenton Lee, and Kristina Toutanova.
\newblock {BERT}: Pre-training of deep bidirectional transformers for language
  understanding.
\newblock In {\em Proceedings of the 2019 Conference of the North American
  Chapter of the Association for Computational Linguistics: Human Language
  Technologies}, pages 4171--4186, 2019.

\bibitem{DiwanGouMusco:2024}
Haya Diwan, Jinrui Gou, Cameron Musco, Christopher Musco, and Torsten Suel.
\newblock Navigable graphs for high-dimensional nearest neighbor search:
  Constructions and limits.
\newblock In {\em \NIPS{2024}}, 2024.

\bibitem{FormalPiwowarskiClinchant:2021}
Thibault Formal, Benjamin Piwowarski, and St\'{e}phane Clinchant.
\newblock {SPLADE:} sparse lexical and expansion model for first stage ranking.
\newblock In {\em Proceedings of the 44th International ACM SIGIR Conference on
  Research and Development in Information Retrieval}, 2021.

\bibitem{fu2016EFANNA}
Cong Fu and Deng Cai.
\newblock {EFANNA}: An extremely fast approximate nearest neighbor search
  algorithm based on {kNN} graph.
\newblock {\em \arXiv{1609.07228}}, 2016.

\bibitem{FuXiangWang:2019}
Cong Fu, Chao Xiang, Changxu Wang, and Deng Cai.
\newblock Fast approximate nearest neighbor search with the navigating
  spreading-out graph.
\newblock {\em Proceedings of the VLDB Endowment}, 12(5):461--474, 2019.
\newblock Data acccessed at: \url{https://github.com/ZJULearning/nsg}.

\bibitem{GarciaDebreuveBarlaud:2008}
Vincent Garcia, Eric Debreuve, and Michel Barlaud.
\newblock Fast k nearest neighbor search using {GPU}.
\newblock In {\em IEEE Computer Society Conference on Computer Vision and
  Pattern Recognition Workshops}, 2008.

\bibitem{GollapudiKariaSivashankar:2023}
Siddharth Gollapudi, Neel Karia, Varun Sivashankar, Ravishankar Krishnaswamy,
  Nikit Begwani, Swapnil Raz, Yiyong Lin, Yin Zhang, Neelam Mahapatro,
  Premkumar Srinivasan, Amit Singh, and Harsha~Vardhan Simhadri.
\newblock {Filtered-DiskANN}: Graph algorithms for approximate nearest neighbor
  search with filters.
\newblock In {\em \WWW{2023}}, pages 3406--3416, 2023.

\bibitem{IndykMotwani:1998}
Piotr Indyk and Rajeev Motwani.
\newblock Approximate nearest neighbors: towards removing the curse of
  dimensionality.
\newblock In {\em \STOC{1998}}, 1998.

\bibitem{IndykXu:2023}
Piotr Indyk and Haike Xu.
\newblock Worst-case performance of popular approximate nearest neighbor search
  implementations: Guarantees and limitations.
\newblock In {\em \NIPS{2023}}, 2023.

\bibitem{JayaramDhulipalaHadian:2024}
Rajesh Jayaram, Laxman Dhulipala, Majid Hadian, Jason Lee, and Vahab Mirrokni.
\newblock {MUVERA}: Multi-vector retrieval via fixed dimensional encoding.
\newblock In {\em \NIPS{2024}}, 2024.

\bibitem{JegouDouzeSchmid:2011}
Herve J{\'e}gou, Matthijs Douze, and Cordelia Schmid.
\newblock Product quantization for nearest neighbor search.
\newblock {\em IEEE Transactions on Pattern Analysis and Machine Intelligence},
  33(1):117--128, 2011.
\newblock Data accessed at \url{http://corpus-texmex.irisa.fr}.

\bibitem{JohnsonDouzeJegou:2021}
Jeff Johnson, Matthijs Douze, and Herv{\'e} J{\'e}gou.
\newblock Billion-scale similarity search with {GPUs}.
\newblock {\em IEEE Transactions on Big Data}, 7(03):535--547, 2021.

\bibitem{karppa22a}
Matti Karppa, Martin Aum\"uller, and Rasmus Pagh.
\newblock {DEANN}: Speeding up kernel-density estimation using approximate
  nearest neighbor search.
\newblock In {\em \AISTATS{2022}}, volume 151, pages 3108--3137, 2022.

\bibitem{KitaevKaiserLevskaya:2020}
Nikita Kitaev, Lukasz Kaiser, and Anselm Levskaya.
\newblock Reformer: The efficient transformer.
\newblock In {\em \ICLR{2020}}, 2020.

\bibitem{Kleinberg:1997}
Jon~M. Kleinberg.
\newblock Two algorithms for nearest-neighbor search in high dimensions.
\newblock In {\em \STOC{1997}}, 1997.

\bibitem{Kleinberg:2000b}
Jon~M. Kleinberg.
\newblock Navigation in a small world.
\newblock {\em Nature}, 406(6798):845--845, 2000.

\bibitem{Kleinberg:2000}
Jon~M. Kleinberg.
\newblock The small-world phenomenon: an algorithmic perspective.
\newblock In {\em \STOC{2000}}, 2000.

\bibitem{KrishnaswamyManoharSimhadri:2024}
Ravishankar Krishnaswamy, Magdalen~Dobson Manohar, and Harsha~Vardhan Simhadri.
\newblock The diskann library: Graph-based indices for fast, fresh and filtered
  vector search.
\newblock {\em IEEE Data Eng. Bull.}, 48(3):20--42, 2024.

\bibitem{KushilevitzOstrovskyRabani:1998}
Eyal Kushilevitz, Rafail Ostrovsky, and Yuval Rabani.
\newblock Efficient search for approximate nearest neighbor in high dimensional
  spaces.
\newblock In {\em \STOC{1998}}, 1998.

\bibitem{Laarhoven:2018}
Thijs Laarhoven.
\newblock Graph-based time-space trade-offs for approximate near neighbors.
\newblock In {\em \SoCG{2018}}, 2018.

\bibitem{lecun1998gradient}
Yann LeCun, L{\'e}on Bottou, Yoshua Bengio, and Patrick Haffner.
\newblock Gradient-based learning applied to document recognition.
\newblock {\em Proceedings of the IEEE}, 86(11):2278--2324, 1998.
\newblock Data accessed at \url{https://github.com/erikbern/ann-benchmarks}.

\bibitem{LewisPerezPiktus:2020}
Patrick Lewis, Ethan Perez, Aleksandra Piktus, Fabio Petroni, Vladimir
  Karpukhin, Naman Goyal, Heinrich K{\"u}ttler, Mike Lewis, Wen-tau Yih, Tim
  Rockt{\"a}schel, Sebastian Riedel, and Douwe Kiela.
\newblock Retrieval-augmented generation for knowledge-intensive {NLP} tasks.
\newblock In {\em \NIPS{2020}}, volume~33, pages 9459--9474, 2020.

\bibitem{li:sigmod20}
Conglong Li, Minjia Zhang, David~G Andersen, and Yuxiong He.
\newblock Improving approximate nearest neighbor search through learned
  adaptive early termination.
\newblock In {\em \SIGMOD{2020}}, pages 2539--2554, 2020.

\bibitem{LiYanZhang:2018}
Jinfeng Li, Xiao Yan, Jian Zhang, An~Xu, James Cheng, Jie Liu, Kelvin K.~W. Ng,
  and Ti-chung Cheng.
\newblock A general and efficient querying method for learning to hash.
\newblock In {\em \SIGMOD{2018}}, pages 1333--1347, 2018.
\newblock Data accessed at
  \url{https://www.cse.cuhk.edu.hk/systems/hash/gqr/datasets.html}.

\bibitem{LuoLakshmanShrivastava:2022}
Chen Luo, Vihan Lakshman, Anshumali Shrivastava, Tianyu Cao, Sreyashi Nag,
  Rahul Goutam, Hanqing Lu, Yiwei Song, and Bing Yin.
\newblock {ROSE}: Robust caches for amazon product search.
\newblock
  https://www.amazon.science/publications/rose-robust-caches-for-amazon-product-search,
  2022.

\bibitem{LvJosephsonWang:2007}
Qin Lv, William Josephson, Zhe Wang, Moses Charikar, and Kai Li.
\newblock Multi-probe {LSH}: efficient indexing for high-dimensional similarity
  search.
\newblock In {\em Proceedings of the 33rd International Conference on Very
  Large Data Bases}, pages 950--961, 2007.

\bibitem{MalkovPonomarenkoLogvinov:2014}
Yury Malkov, Alexander Ponomarenko, Andrey Logvinov, and Vladimir Krylov.
\newblock Approximate nearest neighbor algorithm based on navigable small world
  graphs.
\newblock {\em Information Systems}, 45:61--68, 2014.

\bibitem{MalkovYashunin:2020}
Yury. Malkov and D.~A. Yashunin.
\newblock Efficient and robust approximate nearest neighbor search using
  hierarchical navigable small world graphs.
\newblock {\em IEEE Transactions on Pattern Analysis and Machine Intelligence},
  42(4):824--836, 2020.

\bibitem{ManoharKimBlelloch:2025}
Magdalen~Dobson Manohar, Taekseung Kim, and Guy~E. Blelloch.
\newblock Range retrieval with graph-based indices.
\newblock {\em \arXiv{2502.13245}}, 2025.

\bibitem{ManoharShenBlelloch:2024}
Magdalen~Dobson Manohar, Zheqi Shen, Guy Blelloch, Laxman Dhulipala, Yan Gu,
  Harsha~Vardhan Simhadri, and Yihan Sun.
\newblock {ParlayANN}: Scalable and deterministic parallel graph-based
  approximate nearest neighbor search algorithms.
\newblock In {\em Proceedings of the 29th ACM SIGPLAN Annual Symposium on
  Principles and Practice of Parallel Programming (PPoPP)}, pages 270--285,
  2024.

\bibitem{mialon2023augmented}
Gr{\'e}goire Mialon, Roberto Dessi, Maria Lomeli, Christoforos Nalmpantis,
  Ramakanth Pasunuru, Roberta Raileanu, Baptiste Roziere, Timo Schick, Jane
  Dwivedi-Yu, Asli Celikyilmaz, Edouard Grave, Yann LeCun, and Thomas Scialom.
\newblock Augmented language models: a survey.
\newblock {\em Transactions on Machine Learning Research}, 2023.
\newblock Survey Certification.

\bibitem{Milgram:1967}
Stanley Milgram.
\newblock The small world problem.
\newblock {\em Psychology Today}, 2(1):60--67, 1967.

\bibitem{mitra2018introduction}
Bhaskar Mitra, Nick Craswell, et~al.
\newblock An introduction to neural information retrieval.
\newblock {\em Foundations and Trends{\textregistered} in Information
  Retrieval}, 13(1):1--126, 2018.

\bibitem{MujaLowe:2014}
Marius Muja and David~G. Lowe.
\newblock Scalable nearest neighbor algorithms for high dimensional data.
\newblock {\em IEEE Transactions on Pattern Analysis and Machine Intelligence},
  36(11):2227--2240, 2014.

\bibitem{oliva2001gist}
Aude Oliva and Antonio Torralba.
\newblock Modeling the shape of the scene: A holistic representation of the
  spatial envelope.
\newblock {\em International journal of computer vision}, 42:145--175, 2001.
\newblock Data accessed at \url{https://github.com/erikbern/ann-benchmarks}.

\bibitem{PapadopoulosManolopoulos:2005}
Apostolos~N. Papadopoulos and Yannis Manolopoulos.
\newblock {\em Nearest Neighbor Search:: A Database Perspective}.
\newblock Springer Science \& Business Media, 2005.

\bibitem{pennington2014GloVe}
Jeffrey Pennington, Richard Socher, and Christopher~D. Manning.
\newblock {GloVe}: Global vectors for word representation.
\newblock In {\em Proceedings of the 2014 conference on empirical methods in
  natural language processing (EMNLP)}, pages 1532--1543, 2014.
\newblock Data accessed at \url{https://github.com/erikbern/ann-benchmarks}.

\bibitem{ProkhorenkovaShekhovtsov:2020}
Liudmila Prokhorenkova and Aleksandr Shekhovtsov.
\newblock Graph-based nearest neighbor search: From practice to theory.
\newblock In {\em \ICML{2020}}, 2020.

\bibitem{SimhadriAumullerIngber:2024}
Harsha~Vardhan Simhadri, Martin Aum{\"u}ller, Amir Ingber, Matthijs Douze,
  George Williams, Magdalen~Dobson Manohar, Dmitry Baranchuk, Edo Liberty,
  Frank Liu, Ben Landrum, Mazin Karjikar, Laxman Dhulipala, Meng Chen, Yue
  Chen, Rui Ma, Kai Zhang, Yuzheng Cai, Jiayang Shi, Yizhuo Chen, Weiguo Zheng,
  Zihao Wan, Jie Yin, and Ben Huang.
\newblock Results of the {Big ANN: NeurIPS'23} competition.
\newblock {\em \arXiv{2409.17424}}, 2024.

\bibitem{SimhadriWilliamsAumuller:2022}
Harsha~Vardhan Simhadri, George Williams, Martin Aum\"uller, Matthijs Douze,
  Artem Babenko, Dmitry Baranchuk, Qi~Chen, Lucas Hosseini, Ravishankar
  Krishnaswamny, Gopal Srinivasa, Suhas~Jayaram Subramanya, and Jingdong Wang.
\newblock Results of the {NeurIPS'}21 challenge on billion-scale approximate
  nearest neighbor search.
\newblock In {\em Proceedings of the NeurIPS 2021 Competitions and
  Demonstrations Track}, 2022.

\bibitem{SinghSubramanyaKrishnaswamy:2021}
Aditi Singh, Suhas~Jayaram Subramanya, Ravishankar Krishnaswamy, and
  Harsha~Vardhan Simhadri.
\newblock {FreshDiskANN: A Fast and Accurate Graph-Based ANN Index for
  Streaming Similarity Search}.
\newblock {\em \arXiv{2105.09613}}, 2021.

\bibitem{SpringShrivastava:2017}
Ryan Spring and Anshumali Shrivastava.
\newblock Scalable and sustainable deep learning via randomized hashing.
\newblock In {\em \KDD{2017}}, pages 445--454, 2017.

\bibitem{SubramanyaDevvritKadekodi:2019}
Suhas~Jayaram Subramanya, Devvrit, Rohan Kadekodi, Ravishankar Krishaswamy, and
  Harsha~Vardhan Simhadri.
\newblock {DiskANN: Fast} accurate billion-point nearest neighbor search on a
  single node.
\newblock In {\em \NIPS{2019}}, 2019.

\bibitem{Vargas-MunozGoncalvesDias:2019}
Javier {Vargas Mu{\~n}oz}, Marcos~A. Gon{\c c}alves, Zanoni Dias, and Ricardo
  {da S. Torres}.
\newblock Hierarchical clustering-based graphs for large scale approximate
  nearest neighbor search.
\newblock {\em Pattern Recognition}, 96, 2019.

\bibitem{XiongXiongLi:2021}
Lee Xiong, Chenyan Xiong, Ye~Li, Kwok-Fung Tang, Jialin Liu, Paul~N. Bennett,
  Junaid Ahmed, and Arnold Overwijk.
\newblock Approximate nearest neighbor negative contrastive learning for dense
  text retrieval.
\newblock In {\em \ICLR{2021}}, 2021.

\bibitem{XuManoharBernstein:2025}
Haike Xu, Magdalen~Dobson Manohar, Philip~A. Bernstein, Badrish Chandramouli,
  Richard Wen, and Harsha~Vardhan Simhadri.
\newblock In-place updates of a graph index for streaming approximate nearest
  neighbor search.
\newblock {\em \arXiv{2502.13826}}, 2025.

\bibitem{XuSilwalIndyk:2024}
Haike Xu, Sandeep Silwal, and Piotr Indyk.
\newblock A bi-metric framework for fast similarity search.
\newblock {\em \arXiv{2406.02891}}, 2024.

\bibitem{graspWSDM2022}
Minjia Zhang, Wenhan Wang, and Yuxiong He.
\newblock {GraSP}: Optimizing graph-based nearest neighbor search with subgraph
  sampling and pruning.
\newblock In {\em \WSDM{2022}}, pages 1395--1405, 2022.

\bibitem{ZhaoTianHuang:2023}
Xi~Zhao, Yao Tian, Kai Huang, Bolong Zheng, and Xiaofang Zhou.
\newblock Towards efficient index construction and approximate nearest neighbor
  search in high-dimensional spaces.
\newblock {\em Proceedings of the VLDB Endowment}, 16(8):1979--1991, 2023.

\end{thebibliography}


\clearpage

\appendix

\section{Additional Proofs}\label{sec:appendixProof}



\subsection{Nonexistence of Sparse $\alpha$-Shortcut Reachable Graphs}\label{sec:alphaFailure}

Recent work of Indyk and Xu \cite{IndykXu:2023} shows that, for $k=1$, standard greedy search (i.e., beam search with beam width $b = 1$) provably returns a $\left(\frac{\alpha+1}{\alpha-1} + \epsilon\right)$-approximate nearest neighbor for any constant $\epsilon$ when run on an \emph{$\alpha$-shortcut reachable search graph $G$}. The $\alpha$-shortcut reachability property requires that, for any nodes $x,y \in \{1, \ldots, n\}$ with $d(x,y) > 0$, there is some $z\in \mathcal{N}_G(x)$ with $\alpha \cdot d(z,y) < d(x,y)$ for some parameter $\alpha \ge 1$. The requirement exactly corresponds to navigability (\Cref{def:nav}) when $\alpha = 1$ and is a strictly stronger condition when $\alpha > 1$. 

The guarantee of \cite{IndykXu:2023} is non-vacuous when $\alpha > 1$. 
Unfortunately, it is also not hard to see that for any fixed $\alpha > 1$, there exist relatively low-dimensional point sets with no sparse $\alpha$-shortcut reachable graphs. 
In fact, for any constant $\alpha > 1$, it suffices to consider a random point set in $O(\log n)$ dimensional Euclidean space. 
This contrasts the situation for navigability ($\alpha =1$), since \cite{DiwanGouMusco:2024} show that an $O(\sqrt{n \log n})$ average degree navigable graph can be efficiently constructed for any point set in any dimension (indeed, in any metric space), under the mild assumption of unique pairwise distances between points (which can be ensured, e.g., by tie-breaking with node id).
Formally:
\begin{claim} For any $\alpha > 1$, let $m = O \left (\frac{\log n}{(1-1/\alpha)^2} \right )$. There are $n$ points in $m$-dimensional Euclidean space with unique pairwise distances, but the only $\alpha$-shortcut reachable graph for the points is the complete graph. Further, by \cite{DiwanGouMusco:2024}, the points admit a navigable graph with $O(\sqrt{n \log n})$ ag degree.
\end{claim}
Note that for constant $\alpha > 1$, $1-1/\alpha$ is a constant bounded away from $0$, so $m = O(\log n)$.
\begin{proof}
It suffices to find a set of $n$ points whose pairwise distances all lie in the range $(1/\alpha,1]$. Then, for any $x \neq y$, the only $z$ with $\alpha \cdot d(z,y) < d(x,y)$ is $z = y$. Thus, to ensure $\alpha$-shortcut reachability, all nodes must be connected to all other nodes -- i.e., $G$ must be the complete graph.

If we are not concerned about the dimensionality, finding a set of points in Euclidean space  with all pairwise distances lying in $(1/\alpha,1]$ is trivial: take the $n$ standard basis vectors in $\R^n$, scaled by $1/\sqrt{2}$ so that they all have distance $1$ from each other. Subtract an infinitesimally small random amount from the non-zero entry of each so that all pairwise distances are unique, but still lie in $(1/\alpha,1]$. 

To obtain a result in lower dimensions, we instead consider random points. 
Concretely, consider $n$ points in $\R^{m}$ with each entry set independently to $1$ or $-1$ with probability $1/2$. 
For each $x,y$, we have $\E[\|x-y\|_2^2] = 2m$ and by a starndard binomial 
concentration bounds, $\Pr[|\|x -y \|_2^2 - 2m| \ge m (1-1/\alpha)] \le \exp(-\Omega((1-1/\alpha)^2 \cdot m))$.  Setting $m = O \left (\frac{\log n}{(1-1/\alpha)^2} \right )$, this probability is bounded by $1/n^{c}$ for a large constant $c$. Taking a union bound over all ${n \choose 2} < n^2$ pairs of points, we see that all their squared pairwise distances lie in the range $\left (2m(1-\frac{1-1/\alpha}{2}), 2m (1+\frac{1-1/\alpha}{2}) \right )$ with probability at least $1-1/n^{c-2}$. Normalizing by $ 2m (1+\frac{1-1/\alpha}{2})$, all the squared pairwise distances are less than one $1$ and greater than 
$\frac{1-\frac{1-1/\alpha}{2}}{1+\frac{1-1/\alpha}{2}} \ge 1 - (1-1/\alpha) = 1/\alpha,$ 
where we use the fact that $\frac{1-x}{1+x} \ge 1-2x$ for all $x$. Thus, all squared pairwise distances, and in turn all pairwise distances, lie in the range $(1/\alpha,1)$, as desired. We can again ensure unique pairwise distances by adding arbitrarily small random perturbations to each point, completing the claim.
\end{proof}

\subsection{Failure of Beam Search on Navigable Graphs}\label{sec:beamFailure}

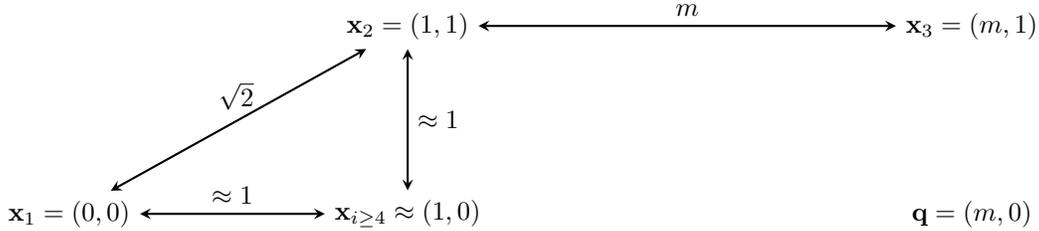
\begin{figure}[h!]
\centering
\begin{tikzpicture}[<->,>=stealth,node distance=2cm,thick]
    \tikzset{vertex/.style = {shape=circle,draw,minimum size=1.5em}}
    
    \node (A) at (-2,0) {$\bv x_1 = (0,0)$};
    \node (B) at (2.5,0) {$\bv x_{i \ge 4} \approx (1,0)$};
    \node (C) at (2.5,2.5) {$\bv x_2 = (1,1)$};
    \node (D) at (10,2.5) {$\bv x_3 = (m,1)$};
    \node (E) at (10,0) {$\bv q = (m,0)$};
    
    
    \draw (A) -- node[above]{$\approx 1$} (B);
    \draw (A) -- node[above]{$\sqrt{2}\,$} (C);
    \draw (B) -- node[right]{$\approx 1$} (C);
    \draw (C) -- node[above]{$m$} (D);
    
\end{tikzpicture}
\caption{Example showing that standard beam search fails to find a nearest neighbor in a navigable graph. Points $\bv x_4,\ldots, \bv x_n$ are all located arbitrarily close to $(1,0)$. They are all  connected to $\bv x_1$ and $\bv x_2$, as well as to each other. The graph is navigable, since we can navigate from $\bv x_1, \bv x_4,\ldots, \bv x_n$ to $\bv x_3$ and vice-versa through $\bv x_2$. All other nodes are directly connected to each other. Suppose beam search with beam width $b \le n-3$ is initialized at $\bv x_1$ with query $\bv q$. Because $\bv x_4,\ldots, \bv x_n$ are all closer to the $\bv q$ than $\bv x_2$, the method will never expand $\bv x_2$ and thus fail to reach the nearest neighbor $\bv x_3$.}
\label{fig:counterExample}
\end{figure}

We next give a simple counterexample, showing that, unless the beam width is set to essentially the full dataset size, standard beam search on a navigable graph can fail to find an approximate nearest neighbor when run on a navigable graph. This observation in part motivates the definition of our alternative ``distance based'' stopping rule, \eqref{eq:our_rule}, and the resulting \ourmethod{} algorithm.
\begin{claim}
For any finite $C$, there exists a set of $n$ points in $2$-dimensional Euclidean space and a navigable graph $G$ such that, for some query point $q$, beam search run on $G$ with beam width $b \le n-3$ returns $\tilde{x}$ with $d(q,\tilde{x}) \geq C\cdot d( q,  x^*)$. 
\end{claim}
\begin{proof}
Consider the following dataset in $2$-dimensional Euclidean space shown in \Cref{fig:counterExample}: $\bv x_1 = (0,0), \bv x_2 = (1,1), \bv x_3 = (m,1)$ for some arbitrarily large value $m$. Let $\bv x_4,\ldots, \bv {x}_n$ all be located at arbitrary positions in an $\epsilon$ ball around $(1,0)$ for arbitrarily small $\epsilon$. We can check that the graph with the following two-way edges is navigable: $(\bv x_2, \bv x_3)$ and $(\bv x_i, \bv x_j)$ for all $i\in \{1,2\}, j\in  \{4,\ldots,n\}$.

Consider beam search initialized at starting point $\bv x_1 = (0,0)$ with query $\bv q = (m,0)$. The nearest neighbor to $\bv q$ is $\bv x_3$ with $\|{\bv q - \bv x_3}\|_2 = 1$. In the first step of beam search, all neighbors of $\bv x_1$ ($\bv x_2,\bv x_4,\ldots, \bv x_n$) will be added to the search queue.  Since $\bv x_2$ is further from $\bv q$ than all nodes in $\bv x_4,\ldots, \bv x_n$, the algorithm will then expand nodes from this set in succession, adding no new nodes to the queue since none of these nodes are connected to $\bv x_3$, the only remaining unexplored node. If $b \le n-3$, the algorithm will then terminate, with $\bv x_2$ never expanded and $\bv x_3$ never explored. 

As a result, beam search returns some $\bv{\tilde x} \in \{\bv x_4,\ldots,\bv x_n \}$ with distance $\|\bv q - \bv{\tilde x}\|_2 \ge m - \epsilon$. It thus achieves approximation factor $\frac{\|\bv q - \bv{\tilde x}\|_2}{\|\bv q - \bv{x}_3\|_2} \ge \frac{m-\epsilon}{1}$. Setting $m = C+\epsilon$ gives the result.
\end{proof}

\section{Additional Implementation Details}

\subsection{Pseudocode for Generalized Beam Search Variants}\label{sec:pseudo}

Below, we provide detailed pseudocode for generalized beam search (\Cref{alg:gen_beam_search}) under stopping conditions  \eqref{eq:greedy_rule} (classic greedy search), \eqref{eq:beam_rule} (classic beam search), and \eqref{eq:our_rule} (\ourmethod{}). While the greedy search order and stopping rule determine the number of distance computations performed, it is possible to optimize runtime and storage requirements by using appropriate data structures to implement the stopping rule. Additionally, we can avoid adding nodes to the candidate set $\mathcal C$ if we are sure that, if popped off $\mathcal{C}$, those nodes would trigger the termination condition anyways.

\smallskip

\noindent\textbf{\ourmethod{} and Greedy Search.} Pseudocode for \ourmethod{} is given in \Cref{alg: alpha approximation}. The same pseudocode can be used for greedy search, by setting the approximation parameter $\gamma = 0$, so that the \ourmethod{} stopping rule \eqref{eq:our_rule} becomes the greedy rule \eqref{eq:greedy_rule}. 

The key optimization is that we maintain a heap, $\mathcal{B}$, of the $k$ nearest points seen so far, which avoids having to extract these neighbors from the set of discovered nodes $\mathcal{D}$ every time termination condition \eqref{eq:our_rule} is checked. Further, if a newly discovered node has distance larger than $(1+\gamma)$ times the $k^\text{th}$ closest seen so far, it will always trigger termination if considered for expansion. Thus, we can avoid adding it to the candidate set of unexpanded nodes, $\mathcal{C}$. See Lines 12-17. This optimization avoids letting $\mathcal{C}$ grow unnecessarily large with nodes that will never be expanded.

\smallskip

\noindent\textbf{Classic Beam Search.} Pseudocode for classic beam search is given in \Cref{alg:classic_beam_search}. The implementation is essentially identical to that of \ourmethod{}, except that a heap of the $b \ge k$ nearest points seen so far must be maintained to efficiently check stopping condition \eqref{eq:beam_rule} each time a node is considered for expansion or newly discovered. At the end of the algorithm, the $k$ nearest points from this heap are ultimately returned. See Lines 22-23.

\begin{algorithm}[h!]
\caption{Adaptive Beam Search}
\label{alg: alpha approximation}
\begin{algorithmic}[1]
\Require Search graph $G$ over nodes $\{1, \ldots, n\}$, starting node $s$, distance function $d$, query $q$, target number of nearest neighbors $k$, approximation parameter $\gamma$.
\Ensure A set of $k$ nodes ${\mathcal{B}} \subset \{1, \ldots, n\}$. Each $x\in {\mathcal{B}}$ is ideally close to $q$ with respect to distance function $d$.
  \algrule

\State $\mathcal{D} \gets \{s\}$ \Comment{Dictionary of Discovered nodes}
\State $\mathcal{C} \gets \{(s, d(q, s))\}$ \Comment{Min-heap of candidates}
\State $\mathcal{B} \gets \{(s, d(q, s))\}$ \Comment{Max-heap of best results}
\While{$C$ is not empty}
    \State $(x, d(q, x)) \leftarrow \text{heappop}(C)$
    \If{$|\mathcal{B}| = k \text{ and }  (1 + \gamma) \cdot \text{findmax}(\mathcal B) \le d(q, x)$} 
        \State \textbf{break} \Comment{Termination condition from Eq.~\eqref{eq:our_rule}}
    \EndIf
    \For{all $y \in \mathcal{N}_G(x)$}
        \If{$y \notin \mathcal{D}$}
            \State $\mathcal{D} \leftarrow \textrm{insert}(\mathcal{D},y)$
            \If{$|\mathcal{B}| < k$ or $d(q, y) < (1+\gamma) \cdot \text{findmax}(\mathcal B)$}
                \State \textrm{heappush}($\mathcal{B}, (y, d(q, y))$)
                \State \textrm{heappush}($\mathcal{C}, (y, d(q, y))$)
                \If{$|\mathcal{B}| = k+1$}
                    \State \textrm{heappop}($\mathcal{B}$)
                \EndIf
            \EndIf
        \EndIf
    \EndFor
\EndWhile
\State \Return $\mathcal{B}$

\end{algorithmic}
\end{algorithm}

\begin{algorithm}[h!]
\caption{Classic Beam Search}
\label{alg:classic_beam_search}
\begin{algorithmic}[1]
\Require Search graph $G$ over nodes $\{1, \ldots, n\}$, starting node $s$, distance function $d$, query $q$, beam width $b$, target number of nearest-neighbors $k$.
\Ensure A set of $k$ nodes ${\mathcal{B}} \subset \{1, \ldots, n\}$. Each $x\in {\mathcal{B}}$ is ideally close to $q$ with respect to distance function $d$.
  \algrule
\State $\mathcal{D} \gets \{s\}$ \Comment{Dictionary of discovered nodes}
\State $\mathcal{C} \gets \{(s, d(q, s))\}$ \Comment{Min-heap of candidates}
\State $\mathcal{B} \gets \{(s, d(q, s))\}$ \Comment{Max-heap of best results}
\While{$\mathcal{C}$ is not empty}
    \State $(x, d(q, x)) \gets \textrm{heappop}(\mathcal{C})$
    \If{$|\mathcal{B}| = b$ and $\textbf{findmax}(\mathcal B) \leq d(q, x)$}
        \State \textbf{break} \Comment{Termination condition from Eq.~\eqref{eq:beam_rule}}
    \EndIf
    \For{all $y \in \mathcal{N}_G(x)$}
        \If{$y \notin \mathcal{D}$}
           \State $\mathcal{D} \gets \textrm{insert}(\mathcal{D},y)$
            \If{$|\mathcal B| < b$ or $d(q, y) < \text{findmax}(\mathcal B)$}
                \State \textrm{heappush}($\mathcal{B}, (y, d(q, y))$)
                \State \textrm{heappush}($\mathcal{C}, (y, d(q, y))$)
                \If{$|\mathcal{B}| = b+1$}
                    \State \textrm{heappop}($\mathcal{B}$)
                \EndIf
            \EndIf
        \EndIf
    \EndFor
\EndWhile
\For{$i=1\ldots(b-k)$} \State $\text{heappop}(\mathcal B)$. \Comment{Reduce $\mathcal B$ down to the best $k$ results. \EndFor}
\State \Return $\mathcal{B}$
\end{algorithmic}
\end{algorithm}

\subsection{Sparse Navigable Graph Construction via Pruning}\label{sec:prune}

As discussed, in \Cref{sec:experiments}, we evaluate the performance of our \ourmethod{} method on both truly navigable graphs, where it is backed by the theoretical guarantee of \Cref{thm:main}, and on heuristic ``approximately navigable'' graphs constructed using a variety of popular methods.

To construct sparse navigable graphs, we use the construction of \cite{DiwanGouMusco:2024}. For $m = \lfloor \sqrt{3n\ln n} \rfloor$, each node is connected to its $m$ nearest neighbors along with $\lceil \frac{3n \ln n}{m}\rceil$ uniformly random nodes. As shown in \cite{DiwanGouMusco:2024}, such a graph is navigable with high probability and has average degree $O(\sqrt{n \log n})$.

We further sparsify these graphs, both to facilitate running large scale experiments and to more accurately reflect performance on graphs with practical levels of sparsity. To do so, we employ a pruning strategy that removes redundant edges from the graph  while maintaining navigability. Pseudocode for the pruning method is given in  \Cref{alg:prune_navigable_graph}. It starts with a navigable graph $G$, then iterates over each node $s$ in the graph, only keeping a minimal set of out edges needed to ensure navigability. In particular, for each node $t \in \{1,\ldots,n\} \setminus \{s\}$, by  \Cref{def:nav}, we must ensure that $s$ has an out neighbor $x$ with $d(x,t) < d(s,t)$. The method iterates over each  $t$, adding an out neighbor of $s$ to the \emph{keep} set only if it is needed to ensure this condition holds for some $t$ (i.e., if no edges alreadt in \emph{keep} ensure the condition). After checking all $t$, it removes all neighbors of $s$ not in \emph{keep}.

\begin{table}[h]
    \caption{Average out degrees of navigable graphs before and after pruning. Note that we run on subsamples of the full datasets from \Cref{tab:datasets} due to the high computational cost of pruning.}
    \label{tab:navigable_graph_degree}
    \centering
    \smallskip
    \begin{tabular}{lcccc}
        \toprule
        \textbf{Dataset} & \textbf{Dimensions} & \textbf{\# Nodes} & \makecell{ \textbf{Average Out Degree} \\ \textbf{Before Pruning}} & \makecell{ \textbf{Average Out Degree} \\ \textbf{After Pruning}} \\
        \midrule
        SIFT1M\cite{JegouDouzeSchmid:2011}  & 128  & 100K  & 3682 & 59 \\
        
        DEEP96 \cite{babenko2016efficient, AumullerBernhardssonFaithfull:2020}  & 96   & 100K & 3682  &77\\
        MNIST \cite{lecun1998gradient, AumullerBernhardssonFaithfull:2020}  & 784  & 50K  & 2516 & 45  \\
    \bottomrule
    \end{tabular}
\end{table}

The pruning strategy can produce navigable graphs that are significantly sparser than those constructed by \cite{DiwanGouMusco:2024}. See \Cref{tab:navigable_graph_degree} for a summary of the average degrees achieved for our tested datasets. Unfortunately, the runtime of our pruning method scales at least quadratically with $n$. This limits our ability to apply the method to the full datasets. An interesting open question is to improve the running time of constructing very sparse and truly navigable graphs.

\begin{algorithm}[h!]
\caption{Navigable Graph Pruning}
\begin{algorithmic}[1]
\Require Navigable graph $G$ on nodes $\{1,\ldots,n\}$, distance function $d$.
\Ensure Subgraph of $G$ that is still navigable over $\{1,\ldots,n\}$ but ideally has many fewer edges.
  \algrule
\For {all $s$ in $\{1,\ldots, n\}$}
    \State $\textrm{keep} \gets \{\}$ \Comment{Set of out neighbors that will remain after pruning.}
    \State $\textrm{remove} \gets \mathcal{N}_G(s)$  \Comment{Set of out neighbors that will be removed after pruning.}
    \For {all $t$ in $\{1,\ldots,n\} \setminus \{s\}$}
        \State $\textrm{navigable} \gets$ FALSE
        \For {all $x$ in $\textrm{keep}$}
            \If {$d(x, t) < d(s, t)$}
                \State $\textrm{navigable} \gets$ TRUE \Comment{Navigability condition satisfied. No need to add an edge.}
                \State \textbf{break}
            \EndIf
        \EndFor
        \If {not $\textrm{navigable}$}
            \For {all $y$ in $\textrm{remove}$}
                \If {$d(y, t) < d(s, t)$}
                    \State $\textrm{keep} \gets \textrm{keep} \cup \{y\}$ \Comment{Keeping edge from $s$ to $y$ ensures navigability to $t$}
                    \State $\textrm{remove} \gets \textrm{remove} \setminus \{y\}$
                    \State \textbf{break}
                \EndIf
            \EndFor
        \EndIf
    \EndFor
    \For {all $y$ in $\textrm{remove}$}
        \State $G.\textrm{remove\_edge}(s, y)$
    \EndFor
\EndFor
\Return $G$
\end{algorithmic}
\label{alg:prune_navigable_graph}

\end{algorithm}

\subsection{Omitted Details on Experimental Setup}\label{sec:expDetails}

We next give additional details on the datasets and graphs used to evaluate \ourmethod{}.

\smallskip

\noindent\textbf{Datasets.}
Table \ref{tab:datasets} summarizes the six benchmark datasets used in our experiments. 
The citation for each dataset includes a note listing the URL where we obtained the specific version of the dataset used in our work.
The datasets are available under the following licenses: MIT License (MNSIST), CC0 1.0 Universal (SIFT, GIST), and the Open Data Commons Public Domain Dedication and License (GloVe). We were unable to find license information for Deep96 and Deep256. Both are available in the public domain. 

For DEEP96, we used a one million point pre-sampled dataset from \cite{AumullerBernhardssonFaithfull:2020}, but our $100$K points used for the navigable graph experiments were sampled from the original dataset available at \url{https://github.com/matsui528/deep1b_gt}.
For GloVe, we sampled one million nodes from the original dataset. The GIST data only includes 1K query points by default. To generate 10K query points, in order to match the other benchmarks, we sampled additional query points uniformly at random from the so-called \emph{learning} data points, which are included with GIST for hyperparameter tuning. We did not use this set of points for any other purpose or any parameter turning. 

\smallskip

\noindent\textbf{Graph Parameters.}
As discussed in \Cref{sec:experiments}, we construct heuristic graphs using four common methods: HNSW \cite{MalkovYashunin:2020},  Vamana \cite{SubramanyaDevvritKadekodi:2019}, NSG \cite{FuXiangWang:2019}, and EFANNA \cite{fu2016EFANNA}. We used our own implementations of HNSW and Vamana. Code for NSG is available under an MIT License at \url{https://github.com/ZJULearning/nsg} and for EFANNA under a BSD License at \url{https://github.com/ZJULearning/efanna}. 

The heuristic graph construction algorithms employed take as input various hyperparameters. Settings used for these hyperparameters are given in \Cref{tab:parameters}. For Vamana, we used the same hyperparameters for all datasets, matching those in the original paper \cite{SubramanyaDevvritKadekodi:2019}, which were found to work well for SIFT, DEEP96, and GIST; using the same parameters for the other datasets yielded similarly good results. 
The hyperparameters for EFANNA \cite{fu2016EFANNA} and NSG \cite{FuXiangWang:2019} for SIFT and GIST are taken from authors' repository \cite{FuXiangWang:2019}. The same parameters were also used by \cite{graspWSDM2022} and \cite{SubramanyaDevvritKadekodi:2019}. For NSG and EFANNA with DEEP96, we used the optimal values used by \cite{graspWSDM2022}. For EFANNA with MNIST, DEEP256, and GloVe, we tested them using the two set of hyperparameters- the ones used for SIFT and GIST- and picked the better performing. We did a similar thing for NSG with MNIST, DEEP256, and GloVe. 

\begin{table}[h!]
\caption{Experimental Hyperparameters for Different Datasetsa dn Graph Constructions}\label{tab:parameters}
\smallskip
\centering
\begin{tabular}{cccccc|cc|cccc|ccc}
\toprule
 & \multicolumn{5}{c|}{\textbf{EFANNA}} & \multicolumn{2}{c|}{\textbf{HNSW}} & \multicolumn{4}{c|}{\textbf{NSG}} & \multicolumn{3}{c}{\textbf{Vamana}} \\
\cmidrule(r){2-15} 
 \textbf{Dataset} & K & L & iter & S & R & M & efC & nn & R & L & C & L & R & $\alpha$  \\
\midrule
SIFT1M & 200 & 200 & 10 & 10 & 100 & 14 & 500 & 200 & 50 & 40 & 500 & 125 & 70 & 2\\
DEEP96 & 200 & 200 & 10 & 10 & 100 & 14 & 500 & 200 & 50 & 40 & 500 & 125 & 70 & 2\\
DEEP256 & 200 & 200 & 10 & 10 & 100 & 14 & 500 & 200 & 50 & 40 & 500 & 125 & 70 & 2\\
GloVe &200 & 200 & 10 & 10 & 100 & 16 & 500 & 400 & 70 & 60 & 500 & 125 & 70 & 2\\
GIST & 400& 400& 12 &15 &100 & 24 & 500 & 400 & 70 & 60 & 500 & 125 & 70 & 2\\
MNIST & 200 & 200 & 10 & 10 & 100  & 14 & 500 & 400 & 50 & 40 & 500 & 125 & 70 & 2\\
\bottomrule
\end{tabular}
\end{table}

For HNSW, we used the hyperparameters that \cite{graspWSDM2022} found to be optimal for SIFT, DEEP96, GIST, and GloVe . For HNSW on MNIST and DEEP256, we tested with values of M=14,16,24 and used the best performing on the standard beam search. Since, the authors found the ideal value of efc for SIFT, DEEP96, GIST, and GloVe to be 500, we used this value for DEEP256 and MNIST.

\smallskip
\noindent\textbf{Computational Resources.}
Navigable graphs were constructed using our pruning methods run on a single core of a 3.2GHz Intel Core i9-12900K CPU with access to 128GB of DDR5 4800mhz RAM. To accelerate pruning  and take advantage of available memory, we precomputed all pairwise distances between pairs of points in the dataset. Each graph required several hours to construct. All other experiments were run on a single 2.9GHz Intel(R) Xeon(R) Platinum 8268 CPU with access to 32GM of RAM, although at most 4GB was used for any individual experiment. Producing a single recall/distance computation tradeoff curve requires several hours for each dataset and algorithm.

\section{Additional Experimental Results}\label{sec:additionalExp}

In this section we include additional experimental results. 

\subsection{Navigable Graphs}\label{sec:addNav}

In \Cref{fig:Navigable Graphs K=100} we compare beam search termination conditions on three datasets for $k=100$. The results are similar to those reported in \Cref{fig:Navigable Graphs} for $k = 1$ and $k = 10$, but with less significant gains seen for \ourmethod{} as compared to standard beam search. As for smaller values of $k$, \ourmethod{} V2 underperforms both other methods.

\begin{figure*}[h]  
    \centering
    \includegraphics[width=\linewidth]{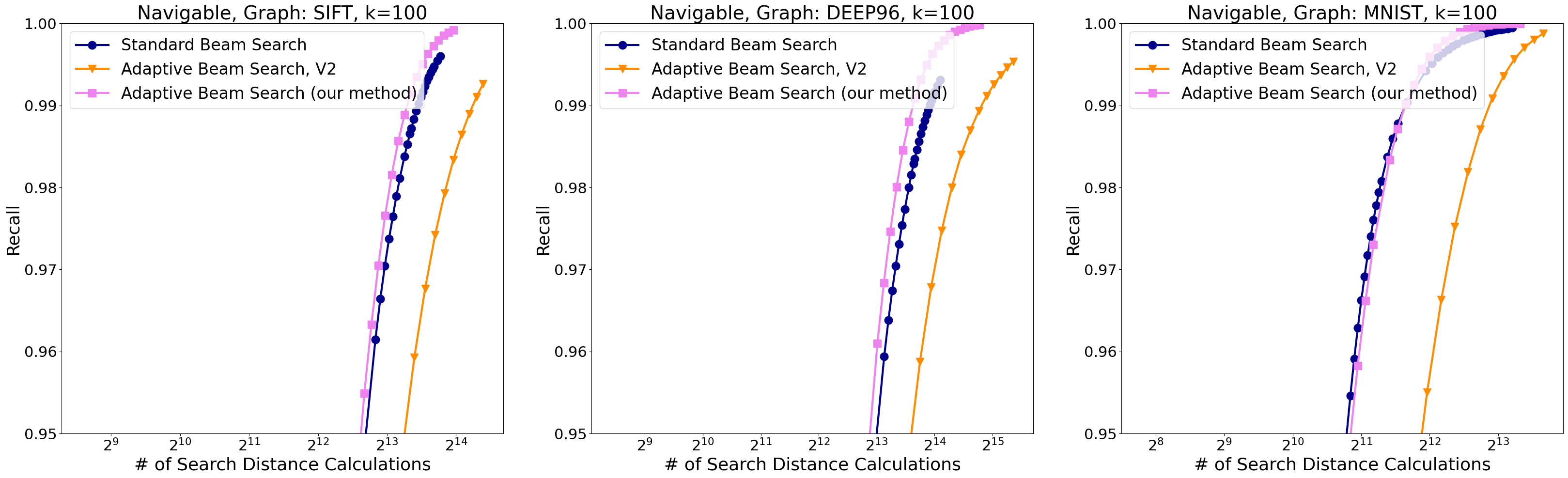}
    \caption{Comparison of generalized beam search termination conditions on navigable graphs across three datasets: SIFT1M, DEEP96, and MNIST (columns), with \( k = 100 \) (rows). \ourmethod{} consistently outperforms standard beam search, while the alternative \ourmethod{} V2 underperforms both by a significant margin.}
    \label{fig:Navigable Graphs K=100}
\end{figure*}

\begin{figure*}[h]  
    \centering
    \includegraphics[width=\linewidth]{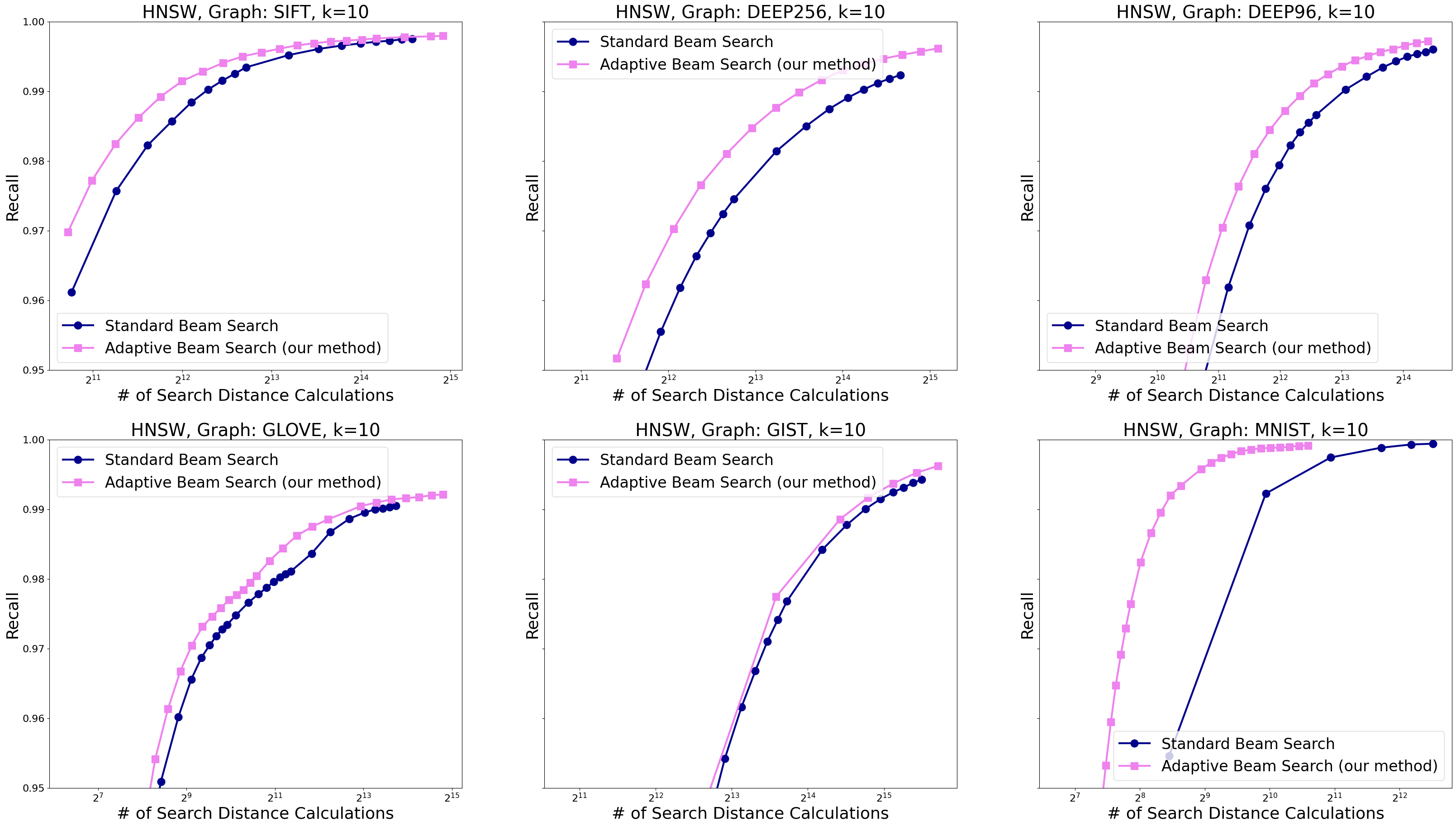}
    \caption{Comparison of generalized beam search termination methods on HNSW graphs with \( k = 10 \) across six datasets. \ourmethod{} outperforms standard beam search, with the degree of improvement varying across datasets.}
    \label{fig:HNSW_k=10}
\end{figure*}

\subsection{Heuristic Graphs}\label{sec:addHeuristic}

In \Cref{fig:HNSW_k=10} we compare beam search termination conditions on HNSW search graphs for all six benchmarks and $k = 10$. In \Cref{fig:HNSW_diff_k} we include further results on HNSW graphs for $k = 1$ and $k = 50$ across three datasets. As with our other experiments on heuristic graphs (see \Cref{fig:K=10_diff_datasets}), we see that \ourmethod{} generally outperforms standard beam search, sometimes by a large margin. One exception is for GIST with $k = 1$, where beam search performs marginally better.

\begin{figure*}[h]  
    \centering
    \includegraphics[width=\linewidth]{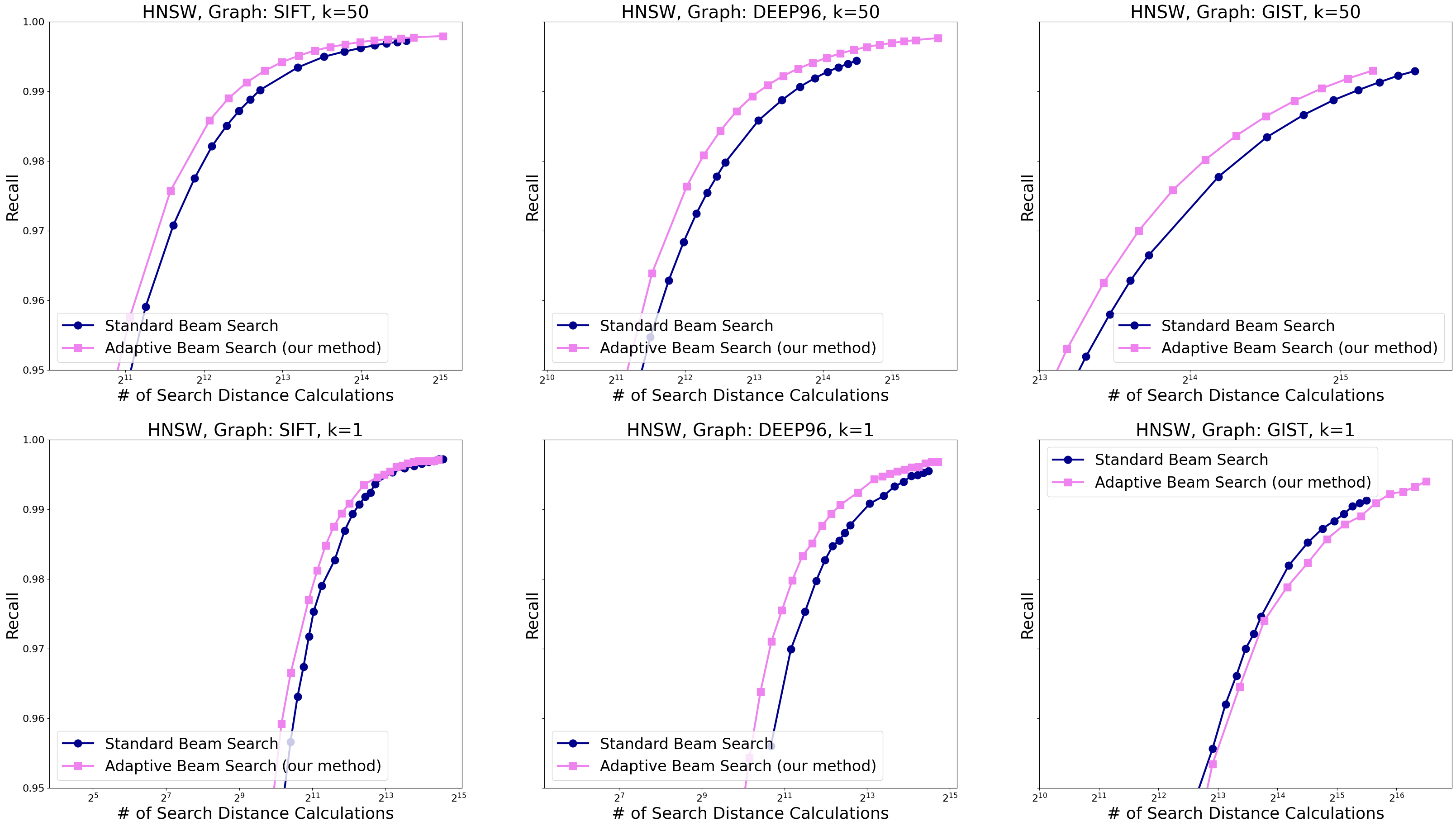}
    \caption{Comparison of generalized beam search termination methods on HNSW graphs across three datasets with \( k = 50 \) and \( k = 1 \). \ourmethod{} outperforms standard beam search as we vary \( k \), with the exception of GIST for \( k = 1 \), where it slightly underperforms.}
    \label{fig:HNSW_diff_k}
\end{figure*}

\begin{figure*}[h]
    \centering
    \begin{subfigure}[b]{0.32\linewidth}
        \centering
        \includegraphics[width=\linewidth]{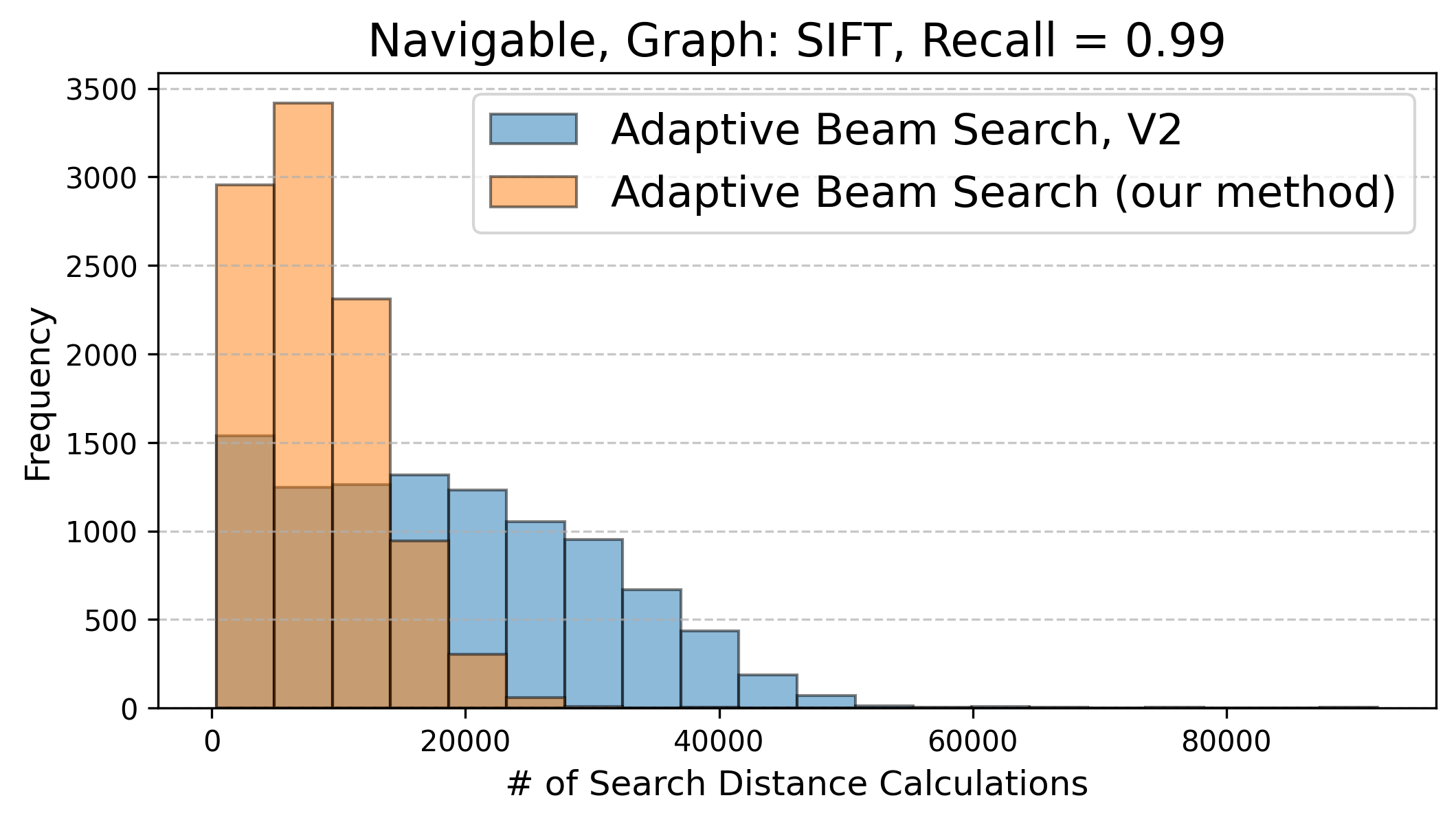}
        \label{fig:navig_histogram_1}
    \end{subfigure}
    \hfill
    \begin{subfigure}[b]{0.32\linewidth}
        \centering
        \includegraphics[width=\linewidth]{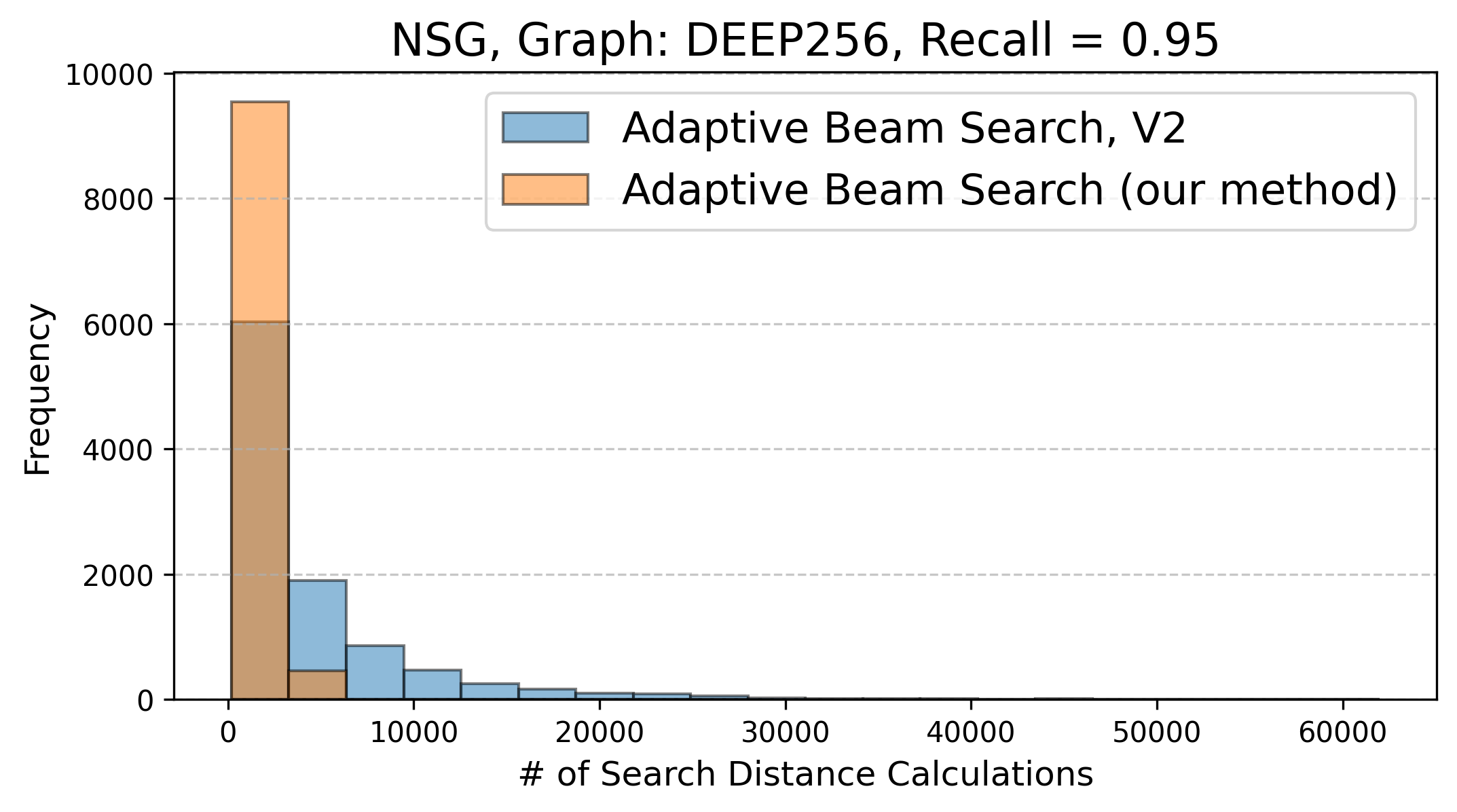}
        \label{fig:nsg_histogram_1}
    \end{subfigure}
    \hfill
    \begin{subfigure}[b]{0.32\linewidth}
        \centering
        \includegraphics[width=\linewidth]{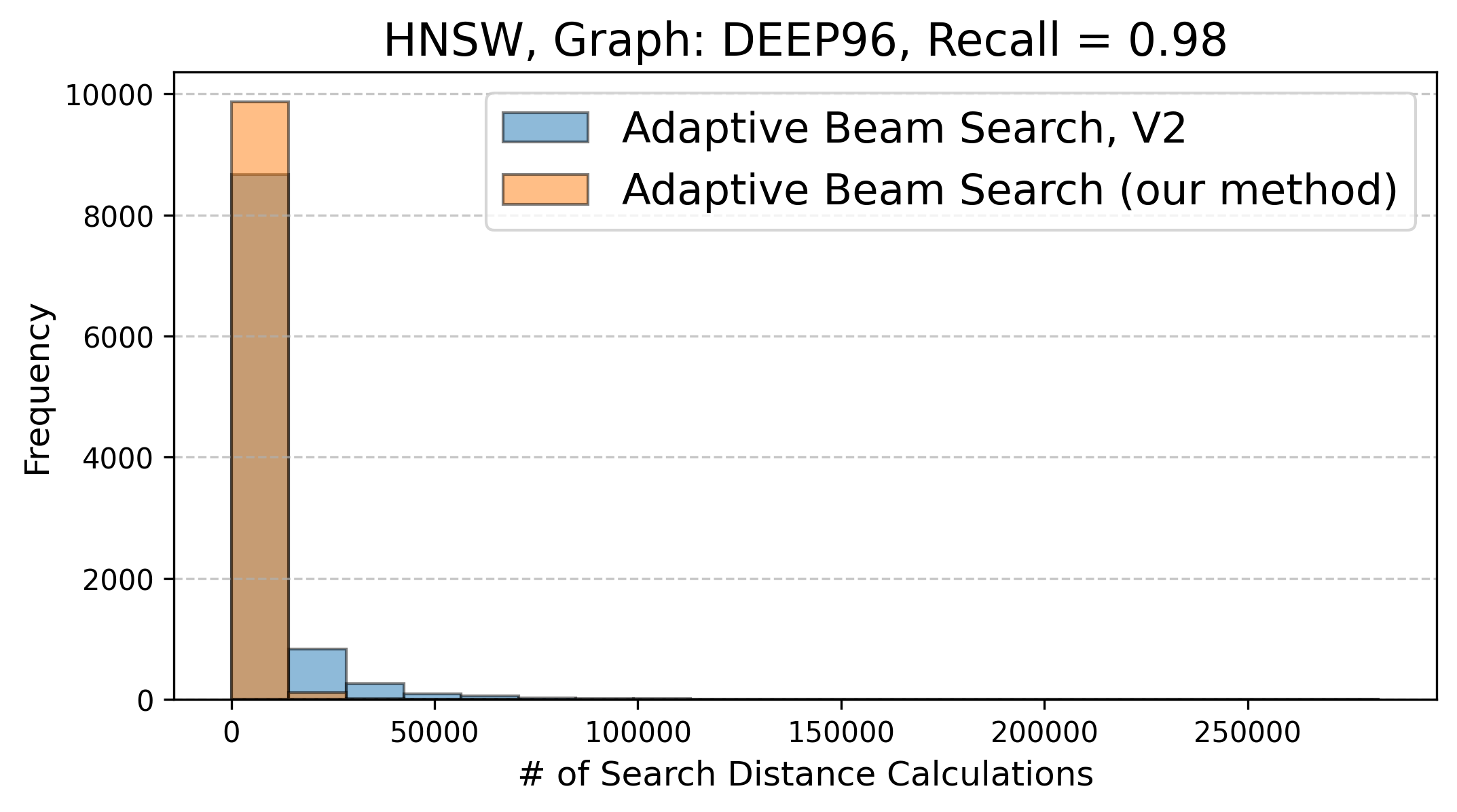}
        \label{fig:hnsw_histogram_1}
    \end{subfigure}
    \caption{Histograms for the number of distance computations performed by \ourmethod{} and \ourmethod{} V2. We tune the $\gamma$ parameter for each method to achieve a fixed recall value, finding that \ourmethod{} V2 has a heavier tail of queries that require many distance computations, in part explaining its poor performance seen in \Cref{fig:Navigable Graphs}.
    }
    \label{fig:v2Histograms}
\end{figure*}

\begin{figure*}[h!]  
    \centering
    \includegraphics[width=\linewidth]{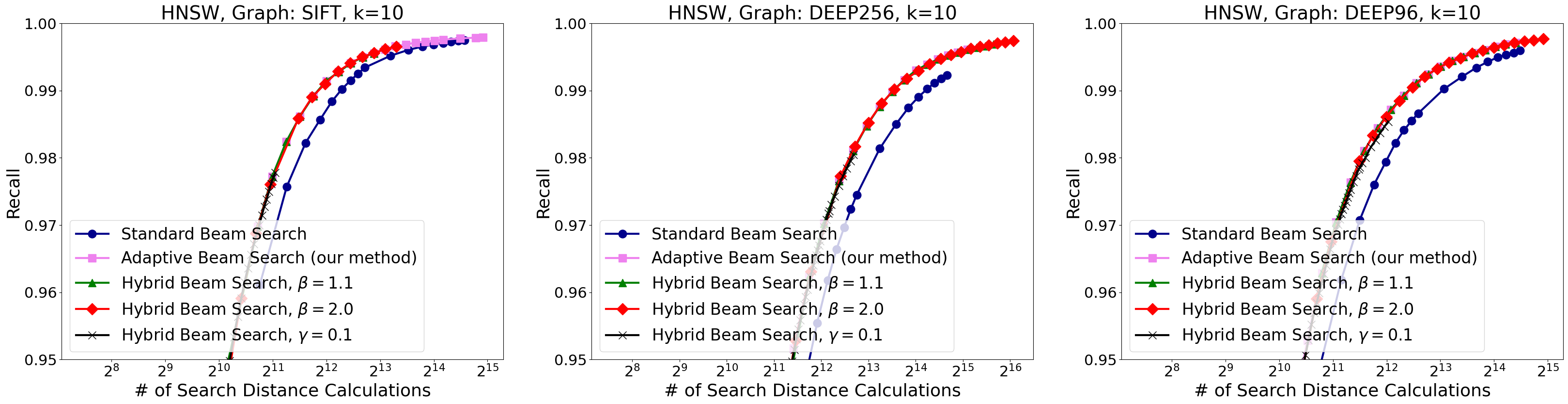}
    \caption{Evaluation of the Hybrid Beam Search termination rule from \eqref{eq:hybrid_rule} on three datasets. There is very little difference in performance between the method and \ourmethod{}. }
    \label{fig:HNSW hybrid}
\end{figure*}

\subsection{\ourmethod{} vs. \ourmethod{} V2}\label{sec:histogram}
As illustrated in \Cref{fig:Navigable Graphs}, \ourmethod{} V2, which uses the more aggressive stopping condition of \eqref{eq:v2}, generally underperforms both \ourmethod{} and classic beam search. We believe this is due to the fact that, to achieve high recall, the $\gamma$ parameter for this rule needs to be set high, causing the method to terminate late and perform a large number of distance computations on some queries. This phenomenon is illustrated in \Cref{fig:v2Histograms}.

\subsection{Hybrid Stopping Rule}
\label{sec:hybrid_rule}
As discussed in \Cref{sec:experiments}, it would be interesting to consider other relaxations of greedy search beyond beam search and \ourmethod{}. One obvious candidate is a rule that combines both relaxations. In particular, in \Cref{alg:gen_beam_search} we could choose to terminate if there are at least:
\begin{align}
    \label{eq:hybrid_rule}
        b \text{ items } j_1, \ldots, j_b\in \mathcal{D} \text{ with } (1+\gamma)  \cdot d(q, j_i) \leq d(q,x), 
\end{align}
where $b > k$ is a ``width parameter'' and $\gamma > 0$ is a distance-based relaxation. We ran initial experiments with this natural hybrid termination, which are shown in \Cref{fig:HNSW hybrid}. To obtain a trade-off curve between recall and distance computations, we either fixed $b = \beta\cdot k$ for a parameter $\beta > 1$ and then varied $\gamma$, or we fixed $\gamma$ and varied $\beta$. Somewhat surprisingly, the hybrid method appears to perform very similarly to \ourmethod{}, although further study of this termination condition and other relaxations would be valuable.


\end{document}